\documentclass[letterpaper, 11pt]{article}
\usepackage[margin=1in]{geometry}
\usepackage[bookmarks, colorlinks=true, plainpages=false,
            citecolor=red,
            linkcolor=blue,
            anchorcolor=red,
            urlcolor=blue]{hyperref}
\usepackage{url}
\usepackage{amssymb, amsmath, mathtools, amsfonts, amsthm, dsfont, mathrsfs}
\usepackage{float, xcolor, xspace, graphicx, appendix, enumitem, booktabs, multirow}
\usepackage{tikz}
\IfFileExists{subfigure.sty}{\usepackage{subfigure}}{}
\IfFileExists{cleveref.sty}{\usepackage{cleveref}}{}
\usepackage{algorithm, algpseudocode}
\usepackage[normalem]{ulem}

\makeatletter
\algrenewcommand\alglinenumber[1]{\footnotesize #1. }
\algrenewcommand\algorithmiccomment[1]{\hfill\textcolor{blue}{$\triangleright$ #1}}
\makeatother

\usepackage{prettyref}

\newrefformat{eq}{(\ref{#1})}
\newrefformat{sec}{Section~\ref{#1}}
\newrefformat{alg}{Algorithm~\ref{#1}}
\newrefformat{fig}{Figure~\ref{#1}}
\newrefformat{tab}{Table~\ref{#1}}
\newrefformat{rmk}{Remark~\ref{#1}}
\newrefformat{clm}{Claim~\ref{#1}}
\newrefformat{claim}{Claim~\ref{#1}}
\newrefformat{def}{Definition~\ref{#1}}
\newrefformat{assump}{Assumption~\ref{#1}}
\newrefformat{cor}{Corollary~\ref{#1}}
\newrefformat{lmm}{Lemma~\ref{#1}}
\newrefformat{prop}{Lemma~\ref{#1}}
\newrefformat{app}{Appendix~\ref{#1}}
\newrefformat{thm}{Theorem~\ref{#1}}

\usepackage[size=tiny]{todonotes}

\newif\ifshowrevisions
\showrevisionstrue

\theoremstyle{plain}
\newtheorem{theorem}{Theorem}
\newtheorem{lemma}{Lemma}
\newtheorem{corollary}{Corollary}
\newtheorem{fact}{Fact}

\theoremstyle{definition}

\newtheorem{remark}{Remark}
\newtheorem{claim}{Claim}

\newif\ifhideproofs
\ifhideproofs
    \usepackage{environ}
    \NewEnviron{hide}{}

\fi
\newif\ifdraft
\drafttrue

\newcommand{\norm}[1]{\left\|{#1} \right\|}

\newcommand{\E}{\mathbb{E}}
\newcommand{\Expect}{\mathbb{E}}

\newcommand{\Var}{\mathsf{Var}}

\makeatletter
\@tfor\@tempa:=ABCDEFGHIJKLMNOPQRSTUVWXYZ\do{
  \expandafter\edef\csname t\@tempa\endcsname{
    {\noexpand\widetilde{\@tempa}}
  }
}
\@tfor\@tempa:=ABCDEFGHIJKLMNOPQRSTUVWXYZ\do{
  \expandafter\edef\csname cal\@tempa\endcsname{
    \noexpand\mathcal{\@tempa}
  }
}
\@tfor\@tempa:=abcdefghijklmnopqrstuvwxyz\do{
  \expandafter\edef\csname sf\@tempa\endcsname{
    \noexpand\mathsf{\@tempa}
  }
}
\makeatother

\newcommand{\Qd}{[0,1]^d}
\newcommand{\Td}{\mathbb T^d}
\newcommand{\Pp}{\mathbb P}

\newcommand{\OPT}{\mathrm{OPT}}
\newcommand{\dist}{\operatorname{dist}}
\newcommand{\NN}{\operatorname{NN}}

\usepackage{color-edits}
\addauthor[Mingwei]{mw}{purple}

\title{Optimal analysis of Greedy for stochastic online Euclidean matching: flat torus to the rescue}
\author{Mingwei Yang and Sophie H.\ Yu\thanks{M. Yang is with the Department of Management Science and Engineering, Stanford University, Stanford CA, USA,
\texttt{mwyang@stanford.edu}.
S.\ H.\ Yu is with The Wharton School of Business, University of Pennsylvania, Philadelphia PA, USA,  \texttt{hysophie@wharton.upenn.edu}.
}}
\date{}

\begin{document}
\maketitle

\begin{abstract}
We study online metric matching with $n$ servers and $n$
requests, all sampled independently and uniformly from $[0,1]^d$.  We analyze the \emph{Greedy} algorithm, which matches each arriving
request to its closest available server.
Despite its strong empirical performance, non-trivial bounds on Greedy's competitive ratio were previously known only for $d = 1$, where it is $O(1)$-competitive~\cite{DBLP:conf/sigecom/BalkanskiFP23}.
We prove that Greedy is $O(1)$-competitive for every fixed $d\geq3$ and has competitive ratio
$\Theta(\sqrt{\log n})$ for $d=2$.
Our proof first analyzes Greedy on the flat torus and then transfers the resulting estimates to the cube.
\end{abstract}

\setcounter{tocdepth}{3}
{\small\tableofcontents}

\section{Introduction}\label{sec:intro}

The \emph{online bipartite matching} problem is initiated by \cite{DBLP:conf/stoc/KarpVV90} and models the allocation of a fixed set of resources
to participants who arrive sequentially.  In this problem, one side of a bipartite graph is
known in advance; each vertex on the other side arrives with its incident
edges, and the algorithm must immediately and irrevocably match it to an unmatched neighbor
or leave it unmatched.  The classical
objective is to maximize the number of matched pairs~\cite{DBLP:conf/stoc/KarpVV90}.
This problem is further generalized to weighted and budgeted settings, which capture online advertising and related
allocation markets~\cite{DBLP:journals/jacm/MehtaSVV07}.
See \cite{DBLP:journals/sigecom/HuangTW24} for a survey on the main models, techniques, and applications of online bipartite matching.

We study the \emph{online metric matching} problem, which is a cost-minimization variant of online bipartite matching.
Here, vertices have locations in a metric space.
We refer to offline vertices as \emph{servers} and online vertices as \emph{requests}, where the cost of matching a request to a server equals their distance.
The objective is to minimize the total matching cost.
Online metric matching provides a stylized model of spatial allocation in ride-hailing and delivery platforms.
Available drivers or couriers correspond to servers, customer requests arrive sequentially, and matching costs represent pickup distances.
Besides, distance can also measure dissimilarity in attributes rather than geographic separation.
For example, in service allocation or online gaming, resources and participants may be represented by vectors of attributes, with distance measuring mismatch in service characteristics or player profiles. When such dissimilarities admit an Euclidean representation, these applications motivate studying online metric matching in dimensions beyond two.

The most natural algorithm for online metric matching is \emph{Greedy}, which matches each request
to its closest unused server.  Greedy is easy to implement and requires no information about future requests, rendering it appealing for real-world applications.  Experiments on real and synthetic data report that Greedy exhibits low matching costs with favorable running time and memory
use~\cite{tong2016online}.
Nevertheless, its theoretical guarantee under worst-case instances is substantially weaker: adversarial instances on
the line can give Greedy an exponential competitive
ratio~\cite{DBLP:journals/jal/KalyanasundaramP93,DBLP:journals/tcs/KhullerMV94} and, when the arrival order is randomly permuted, a polynomial competitive ratio~\cite{DBLP:journals/orl/GairingK19}.

This discrepancy between Greedy's empirical performance and worst-case guarantees motivates analyzing it under more realistic models.
\cite{DBLP:conf/sigecom/AkbarpourALS22} consider the \emph{fully random} model, where servers and requests are all sampled independently and uniformly from $[0,1]$, and establish an $O(\log^3 n)$ competitive ratio when there are $n$ requests and $(1+\varepsilon)n$ servers, for any fixed $\varepsilon>0$.
\cite{DBLP:conf/sigecom/BalkanskiFP23} improve this guarantee to $O(1)$ and additionally prove constant competitiveness in the balanced case of $n$ servers and $n$ requests.
They further study the \emph{semi-random} model with adversarial server locations and iid uniform requests on $[0,1]$, obtaining a tight $\Theta(\log n)$ competitive ratio.
While the performance of Greedy is well-understood on the line, a central question concerns whether these guarantees extend to higher dimensions.

\subsection{Main Results}

As our main results, we determine the tight competitive ratio of Greedy in the balanced fully random model for every fixed dimension.  
The following theorem gives the competitive-ratio
upper bound.

\begin{theorem}\label{thm:costs}
    In the balanced fully random model, Greedy is $O(1)$-competitive for $d \neq 2$ and is $O(\sqrt{\log n})$-competitive for $d = 2$.
\end{theorem}

\Cref{thm:costs} recovers the one-dimensional guarantee
of~\cite[Theorem 1]{DBLP:conf/sigecom/BalkanskiFP23} and establishes
constant competitiveness for $d\ge3$.  In the plane, the following theorem establishes the matching lower bound.

\begin{theorem}\label{thm:lb-2d}
    In the balanced fully random model, Greedy is $\Omega(\sqrt{\log n})$-competitive for $d = 2$.
\end{theorem}

Together, \Cref{thm:costs,thm:lb-2d} determine the competitive ratio of Greedy in
the plane as $\Theta(\sqrt{\log n})$, in contrast to the constant competitive ratio for other dimensions.
Notably, in the balanced fully random model with $d = 2$, there exist other online algorithms that attain a constant competitive ratio~\cite{holden2021gravitational,kanoria2025dynamic}, indicating that Greedy is suboptimal.
The distinction between the plane and other dimensions underlying our results is aligned with the observations in prior work that optimal matchings in the plane exhibit unique structures~\cite{talagrand2022upper,kanoria2025dynamic,DBLP:journals/ior/YangY26,DBLP:conf/innovations/LiV026}.

\subsection{Technical Overview}

As the main difficulty of analyzing Greedy, the statuses of the servers jointly depend on earlier matches, and hence the remaining servers form a dependent point configuration.
Moreover, due to the asymmetry between the servers lying close to the boundary and those away from the boundary, the spatial distribution of each server, conditional on it being available, is non-uniform.

\paragraph{Starting with the flat torus.}
To overcome the above barriers, we depart from prior work by analyzing Greedy first on the \emph{flat torus}, which is obtained by identifying opposite faces of the unit cube.
This removes the boundary and restores translation symmetry: simultaneously translating
all inputs also translates every Greedy choice without changing the input law.
Consequently, the distribution of the remaining-server configuration is invariant under torus translations (Fact~\ref{fact:torus-translation-stationarity}).
This implies that, when $m$ servers remain, the expected number of available servers in any measurable region $P$ equals $m|P|$, where $|P|$ denotes the volume of $P$.

To see how the above translation-invariance property helps, consider the matching cost of Greedy at a distance
scale $r$.  We partition the torus into equal grid cells of side length
roughly $r$, and each cell has diameter roughly $\sqrt{d} r$.
Note that a match longer than this diameter must cross between cells, and it can occur under Greedy only if the request's cell contains no available server.
Hence, upper bounding the total matching cost entails upper bounding the probability of a cell containing no available server.
Specifically, suppose that $m$ servers remain, and let $N_P$ denote the remaining-server count in cell $P$.
Chebyshev's inequality gives
\begin{align}
  \Pp\{N_P=0\}\le\frac{\Var(N_P)}{\E[N_P]^2} = \frac{\Var(N_P)}{(m|P|)^2},
  \label{eqn:intro-chebyshev}
\end{align}
where the equality holds by the translation-invariance property for torus Greedy.

To control the variance term $\Var(N_P)$, we use a well-known stability property of Greedy valid under any fixed
metric: replacing one initial server or one past request changes the remaining-server configuration at any moment by at most one unit, and hence at most two cells have their remaining-server counts change, each by one (\Cref{lem:greedy-stability}).
Consequently, the Efron--Stein inequality (\Cref{lem:efron-stein}) gives
\begin{align}
  \sum_P\Var(N_P)\le 2n-m.
  \label{eqn:intro-var}
\end{align}
Remarkably, this variance bound does not depend on the number of cells and holds for any metric (\Cref{prop:aggregate-variance}).
Combining with \eqref{eqn:intro-chebyshev}, we obtain a tail bound for the matching cost of torus Greedy at scale $r$, yielding an upper bound for the total Greedy cost on torus (\Cref{prop:torus-edge} and Corollary~\ref{cor:torus-mean-costs}).

\paragraph{Transferring torus estimates to the cube.}
To transfer the Greedy estimates on torus to our target metric, the Euclidean metric in the unit cube, we open the torus into a cube by choosing a uniformly random cut position on each coordinate circle, and unwrapping
each circle into $[0,1)$.
Conditional on the cuts, the unwrapped torus inputs remain independent and uniform on the unwrapped torus. 
Consider a new metric induced by Euclidean distances on the unwrapped torus, and we refer to it as the \emph{cut metric}.
Crucially, conditional on the cuts, the Greedy process under the cut metric has precisely the same law as the original Euclidean Greedy process and hence can serve as an intermediate to compare the original Euclidean Greedy process and the torus Greedy process (\Cref{lem:unwrapped-greedy-law}).

We then show that the Greedy process under the cut metric and the torus Greedy process stay close.
To this end, notice that the discrepancy between the remaining-server configurations of the two Greedy processes increases after a new request only when a cut separates this request and
its torus nearest available server along a shortest coordinate arc, and the stability property of Greedy ensures that this increment is at most one.
The probability of such a cut occurring can be upper bounded by the torus matching length of this request (Fact~\ref{lem:random-cut}).
Hence, the overall discrepancy between these two remaining-server configurations can be controlled by the total matching cost of torus Greedy (\Cref{lem:discrepancy}).
By the translation-invariance property of the torus Greedy, the same coupling also bounds the total variation distance between the expected server counting measure of the original Euclidean Greedy process and the uniform measure (\Cref{prop:mean-bias}).

\paragraph{The cube upper bound.}
Using the same second-moment approach as for the torus, we now bound the total matching cost of Euclidean Greedy.
The variance bound \eqref{eqn:intro-var} still holds, but we no longer have a precise estimate on $\E[N_P]$ as in \eqref{eqn:intro-chebyshev}.
To get around this, we classify cells according to their expected remaining-server count $\E[N_P]$.
Specifically, we call a cell $P$ \emph{good} if $\E[N_P] = \Omega(m|P|)$, and
\emph{bad} otherwise.
On good cells, Chebyshev's inequality and the aggregate variance bound
control the probability of being empty just as in \eqref{eqn:intro-chebyshev}, while the total volume
of bad cells is controlled by the aforementioned estimate on the total variation distance between the expected server counting measure of the original Euclidean Greedy process and the uniform measure.
This yields the desired upper bound for Euclidean Greedy.

\paragraph{The planar lower bound.}
To prove the lower bound for $d = 2$, we select $\Theta(\log n)$ disjoint stages, where each stage consists of a time interval, and prove that Greedy incurs expected cost $\Omega(\sqrt n)$ within each
stage.
To analyze a stage, suppose that $m$ servers remain at the beginning of this stage, and we impose a randomly
shifted grid of side length $h_m=\Theta(\sqrt n/m)$.
Define the \emph{deficit} of a grid cell $P$ as the positive difference between $m|P|$, the
expected number of future requests landing in $P$, and its remaining-server count $N_P$.
Intuitively, a cell with a large deficit must have many of its future requests being matched to the servers in other cells, and each such cell-crossing match contributes a cost $\Omega(h_m)$.
We formally show that an expected total deficit $\Omega(m)$ suffices to force $\Omega(m)$ cell-crossing matches in this stage, incurring an expected cost $\Omega(mh_m)=\Omega(\sqrt n)$ (\Cref{prop:delayed-service}).

We then show that a deficit bound for the torus can be transferred to a deficit bound for the square (\Cref{lem:torus-square-deficit}).
Specifically, under the random-cut coupling between the torus Greedy process and the original Euclidean Greedy process developed in the upper-bound proof, the cells split by the cuts only contribute a negligible amount of deficit, and hence we can discard the contribution from these cells.
The deficit transfer is then concluded by the established discrepancy between the remaining-server configurations of these two Greedy processes.

To prove the deficit bound for the torus, recall that the past-request counts of the cells are binomial random variables and hence have expected total absolute deviation of $\Theta(\sqrt{n} / h_m) = \Theta(m)$.
Our proof strategy is to show that the remaining-server counts inherit the fluctuations from the past-request counts.
This inheritance can be made formal under the additional assumption that torus Greedy has small sensitivity to skipping one request (\Cref{prop:demand-localization}).
Specifically, consider two torus Greedy processes: the process $\mathsf G$ that serves the requests $X_1, \ldots, X_{n-m}$, and the process $\mathsf G^{(-u)}$ that serves the requests $X_1, \ldots, X_{u - 1}, X_{u + 1}, \ldots, X_{n - m}$.
The stability property of Greedy ensures
that $\mathsf G^{(-u)}$ has the same remaining servers as $\mathsf G$ plus exactly one extra server, denoted by $Z_{m,u}$.
If $X_u$ and $Z_{m,u}$ lie in the same cell $P$, then restoring $X_u$ increases
the past-request count of $P$ by one and decreases the remaining-server
count of $P$ by one, leaving their sum unchanged.
Consequently, if many $(X_u, Z_{m, u})$ pairs have both points belonging to the same cell, then the fluctuations of the past-request counts can be passed to the remaining-server counts.
Recall that we impose a randomly shifted grid of side length $h_m$, and a standard property of randomly shifted grids bounds the probability of $X_u$ and $Z_{m, u}$ lying in different cells
by $O(\E d_{\mathbb T}(X_u, Z_{m, u}) / h_m)$, where $d_{\mathbb T}$ denotes the torus distance.
The required geometric sensitivity property is therefore an upper bound on $\E d_{\mathbb T}(X_u, Z_{m, u})$.

\paragraph{Geometric sensitivity bound.}
Bounding $\E d_{\mathbb T}(X_u, Z_{m, u})$ constitutes the most technical part of our analysis.
To see the difficulty, skipping one request can change a series of later matches, so Greedy's stability property alone does not control this distance.
Such a chain reaction was previously described as the main technical challenge of analyzing Greedy~\cite{kanoria2025dynamic}.

In our analysis, we fix the remaining-server configuration just before $X_u$ and the server matched to each subsequent request under $\mathsf G^{(-u)}$.
Under this conditioning, each subsequent request is independently uniform
in the Voronoi cell of the server it matches to, while $X_u$ remains
independent and uniform on the entire torus.
Consider the unique available server under $\mathsf G^{(-u)}$ that is unavailable under $\mathsf G$.
This extra server moves as subsequent requests arrive, and it ultimately reaches $Z_{m, u}$.
In a periodic representation of the torus, we characterize the extra server's random motion by two periodic Voronoi properties: the first property establishes that this motion is a martingale, while the second property upper-bounds the expected squared length of the increment of this motion by the increment of some quadratic potential (\Cref{lem:periodic-voronoi-deletion}).
These quadratic-potential increments then telescope, enabling us to upper bound $\E d_{\mathbb T}(X_u, Z_{m, u})$ by the expected cost of torus Greedy for matching the request $X_{n - m}$ (\Cref{prop:endpoint-localization}).

\subsection{Related Literature}\label{sec:related-literature}

\paragraph{Online metric matching with adversarial input.}
Online metric matching has been studied in the adversarial setting,
where the optimal competitive ratio achieved by deterministic algorithms on general metrics is
$2n-1$~\cite{DBLP:journals/jal/KalyanasundaramP93,DBLP:journals/tcs/KhullerMV94}.
\cite{DBLP:conf/soda/MeyersonNP06} propose a randomized algorithm with competitive ratio $O(\log^3 n)$, which is later improved to $O(\log^2 n)$~\cite{DBLP:journals/algorithmica/BansalBGN14}.
For the line, \cite{DBLP:conf/compgeom/Raghvendra18} gives a deterministic $O(\log n)$-competitive algorithm, while an
$\Omega(\sqrt{\log n})$ lower bound holds even for randomized
algorithms~\cite{DBLP:journals/talg/PesericoS23}.
When the request
arrival order is randomly permuted, \cite{DBLP:conf/approx/Raghvendra16} gives a deterministic algorithm that achieves an optimal
$\Theta(\log n)$ competitive ratio on general metrics while retaining an
asymptotically optimal adversarial guarantee.
\cite{DBLP:journals/orl/GairingK19} show that Greedy is $n$-competitive in random-order arrivals
and admits a polynomial competitive-ratio lower bound even on the
line.
In non-bipartite matching on the line with general arrivals, where every
agent arrives online and may wait for a match, \cite{DBLP:journals/corr/abs-2606-05546} show that every algorithm has an unbounded competitive ratio
under random-order
arrivals.

\paragraph{Online metric matching with stochastic input.}
For adversarial servers and iid requests from a known distribution, the algorithm of \cite{DBLP:conf/icalp/GuptaGPW19}
achieves an $O((\log\log\log n)^2)$ competitive ratio on general metrics
and a constant competitive ratio on trees.
\cite{kanoria2025dynamic} presents a constant-competitive algorithm when all inputs are uniformly distributed in $[0, 1]^d$, and considers more general models where servers can outnumber requests, or servers can be replenished.
\cite{DBLP:journals/ior/YangY26} discover a reduction from arbitrary initial server configurations to fully
random input, and then give a constant-competitive algorithm for smooth request distributions on $[0,1]^d$ when $d\ge3$.
\cite{DBLP:conf/innovations/LiV026} adopt the smoothed-analysis framework that further permits independent
requests from different smooth distributions; for fixed
$d\ne2$, they give an $O(1)$-competitive algorithm that uses one sample from each request distribution and no further distributional knowledge.
\cite{kumar2026feature,DBLP:journals/ior/ChenKKZ26} study more general feature-based models that allow different supply and demand distributions and richer match values.
Under the non-bipartite general-arrivals model, \cite{DBLP:journals/corr/abs-2606-05546} give an
$O(\log^2 n)$-competitive algorithm on the line when arrivals are iid sampled from an unknown distribution.

\paragraph{Greedy for online metric matching.}
The competitive ratio of Greedy has been extensively studied on the line for iid uniform servers and requests.
\cite{DBLP:conf/sigecom/AkbarpourALS22} prove an $O(\log^3 n)$ competitive ratio when there are $\Theta(n)$ excess
servers, and \cite{DBLP:conf/sigecom/BalkanskiFP23} give an $O(1)$ competitive ratio both in the balanced case and with a linear excess of
servers.
\cite{DBLP:conf/sigecom/BalkanskiFP23} further establish
a tight $\Theta(\log n)$ competitive ratio for adversarial servers and iid uniform requests.
\cite{DBLP:journals/mor/SentenacNLMP26} obtain exact asymptotic
costs for Greedy with $\Theta(n)$ excess servers.
In higher dimensions, \cite{DBLP:journals/ipl/TsaiTC94} give an $O(\sqrt n)$
competitive ratio for iid uniform input on the unit
disk.

The recent one-dimensional analyses of Greedy heavily exploit the order structure of the
line and hence cannot be easily generalized to higher dimensions. \cite{DBLP:conf/sigecom/AkbarpourALS22}
represent spatially ordered servers and requests by a random walk,
whose exit times partition the line into intervals with a server surplus.
The analysis in~\cite{DBLP:conf/sigecom/BalkanskiFP23} compares
Greedy with the algorithm of \cite{kanoria2025dynamic} through hybrid algorithms that switch
from one algorithm to the other, controlling the gap between the two
servers on which the coupled matching processes differ.

\section{Preliminaries}\label{sec:preliminaries}

Let $Y_1,\ldots,Y_n$ be the servers and $X_1,\ldots,X_n$ be the requests,
all sampled independently and uniformly from the unit cube $\Qd$.
The requests arrive in order.  Upon the arrival of $X_i$, the Greedy algorithm matches
it to the closest available server $Y_j$, incurring a cost $\norm{X_i-Y_j}_2$.
The matched server $Y_j$ then becomes unavailable to all subsequent requests.
Throughout, we assume distinct server locations and unique nearest-server
choices, which hold almost surely for all concerned metrics and resampled instances.

Let $G_n$ denote the total matching cost incurred by Greedy.
Let $\OPT_n$ denote the minimum total cost of any one-to-one
matching between all servers and requests.
We define the \emph{competitive ratio} of Greedy as $\E G_n/\E\OPT_n$, where both expectations are over
the independent server and request samples.

We define the \emph{rank} of a state of the Greedy process to be its number of unmatched servers.
At rank $m$, exactly $m$ servers are unmatched, and the next request moves the
process to rank $m-1$.
For $0\le m\le n$, let $S_m$ denote the set of servers remaining after Greedy has matched
the first $n-m$ requests.
For $m \in [n]$, write
$X^{(m)}=X_{n-m+1}$ for the next request arriving at rank $m$.

We will apply the following two forms of the Efron--Stein inequality.

\begin{lemma}[\cite{steele1986efron,boucheron2005moment}]
\label{lem:efron-stein}
Let $Z=(Z_1,\ldots,Z_N)$ be a random vector with independent coordinates, and let $f(Z)$ be a square-integrable real-valued function.
\begin{enumerate}[label=\textup{(\roman*)}]
\item \textbf{(Resampling)}
Let $Z^{(j)}$ be obtained from $Z$
by replacing $Z_j$ with an independent copy $Z_j'$, leaving all other
coordinates unchanged.
Then
\begin{equation}\label{eq:efron-stein-resampling}
  \Var(f(Z))\le\frac12\sum_{j=1}^{N}
       \E\bigl[(f(Z)-f(Z^{(j)}))^2\bigr].
\end{equation}
\item \textbf{(Omitting One Input)}
For each $j \in [N]$, let $f^{(-j)}(Z)$ be any square-integrable real-valued
function that does not depend on $Z_j$.
Then
\begin{equation}\label{eq:efron-stein-deletion}
  \Var(f(Z))\le\sum_{j=1}^{N}
       \E\bigl[(f(Z)-f^{(-j)}(Z))^2\bigr].
\end{equation}
\end{enumerate}
\end{lemma}

\subsection{The Flat Torus}
\label{sec:torus-preliminaries}

Think of the unit cube with wraparound in every coordinate: crossing
one face brings us back through the opposite face.  This periodic cube is
called the \emph{flat torus}, denoted by $\mathbb T^d$.  We represent its
points by coordinates in $[0,1)^d$ and perform addition and subtraction
coordinatewise modulo one, keeping only the fractional part of each
coordinate.  The conventional notation $\mathbb T^d=\mathbb R^d/\mathbb Z^d$
expresses this same rule: points whose coordinates differ by integers
represent the same torus point.

Distance on the torus allows these wraparound shortcuts.\footnote{For example, in one dimension the torus distance between $0.99$ and $0.01$ is $0.02$,
rather than their Euclidean distance $0.98$.}  For a coordinate difference
$a\in\mathbb R$, its shortest wraparound length is
$|a|_{\mathbb T}:=\min_{k\in\mathbb Z}|a-k|$.
Combining the shortest wraparound lengths in all coordinates gives the
torus distance between $x,y\in[0,1)^d$:
\[
  d_{\mathbb T}(x,y)
  :=\left(\sum_{i=1}^d|x_i-y_i|_{\mathbb T}^2\right)^{1/2}
   =\min_{z\in\mathbb Z^d}\|x-y+z\|_2.
\]
Each coordinate distance is at most $1/2$, and therefore
$\operatorname{diam}(\mathbb T^d)=\sqrt d/2$.

The formula $d_{\mathbb T}(x,y) = \min_{z\in\mathbb Z^d}\|x-y+z\|_2$ also gives an alternative interpretation of the torus:
for each point $y \in [0, 1)^d$, create a \emph{periodic copy}
$y+z$ for every $z \in \mathbb Z^d$.  The torus distance between $x$ and $y$ now equals
the ordinary Euclidean distance from $x$ to the closest periodic copy of
$y$.  We will use both the \emph{periodic-cube viewpoint}, which keeps
coordinates in $[0,1)^d$ and uses wraparound distance, and the
\emph{periodic-lift viewpoint}, which uses all the periodic copies in
$\mathbb R^d$.

The uniform measure on $\mathbb T^d$ is simply the Lebesgue measure on $[0,1)^d$.
We use $|P|$ to denote the volume of a measurable region $P\subseteq\mathbb T^d$.

\paragraph{Translation symmetry and torus stationarity.}
The torus metric is invariant under a common translation of all input
points.  For $u\in\mathbb T^d$, define the torus translation
\[
  \tau_u(x):=x+u\pmod 1.
\]
The following fact states the resulting torus stationarity and its consequence
for expected remaining-server counts under uniform input.

\begin{fact}
\label{fact:torus-translation-stationarity}
\label{item:torus-station}
For Greedy on $n$ iid uniform servers and $n$ iid uniform requests on
$\mathbb T^d$, let $B_m$ be the set
of unmatched servers at rank $m$.  For every $0\le m\le n$, the
distribution of $B_m$ is invariant under torus translations.
Consequently, for every measurable region $P\subseteq\mathbb T^d$,
\begin{equation}
\label{eq:torus-stationarity}
  \E|B_m\cap P|=m|P|.
\end{equation}
\end{fact}

Remarkably, Fact~\ref{fact:torus-translation-stationarity} gives expected remaining-server counts for torus Greedy without asserting independence of the
remaining servers.  The proof of Fact~\ref{fact:torus-translation-stationarity}
is deferred to
Appendix~\ref{app:proof-torus-translation}.

\subsection{Random Cuts and Unwrapping}
\label{sec:random-cuts-unwrapping}

To transfer the torus estimates to the cube, we cut
open the periodic cube to obtain an ordinary cube.  In one dimension,
imagine choosing a point on a circle and opening it there to form an
interval; in higher dimensions, we do this in each coordinate.
In particular, we will choose the cuts uniformly at random.

Independently of all inputs, randomly sample the cut positions $U=(U_1,\ldots,U_d)\sim\operatorname{Unif}(\mathbb T^d)$.
In coordinate $i$, we cut at $U_i$, and the set $\{z:z_i=U_i\}$ is
called the \emph{seam} in this coordinate.  Each seam becomes a pair of
opposite faces of the opened cube.
The new coordinates of a point $z$ in the opened cube are measured
from the cut position, modulo one, using the \emph{unwrapping map}
\[
  \kappa_U(z):=(z-U)\pmod 1\in[0,1)^d,
\]
and we measure the distance in the opened cube using ordinary Euclidean distance:
\[
  d_U(z,z'):=\|\kappa_U(z)-\kappa_U(z')\|_2.
\]
We refer to $d_U$ as the \emph{cut metric}.\footnote{For example, cutting the unit circle at $0$ places $0.99$ and $0.01$ near opposite ends of the interval, increasing their distance from $0.02$ to $0.98$.  Cutting at $0.5$ instead maps them to $0.49$ and $0.51$, so their distance remains $0.02$.}

\phantomsection\label{item:fact-cut-item-iv}
Regard each coordinate modulo one as a point on a circle of
circumference one.
For $z,z'\in\mathbb T^d$, we choose a shortest route along this circle
between $z_i$ and $z_i'$, independently of $U$, and call this route the \emph{chosen shortest coordinate arc}.
Its length is $\Delta_i(z,z'):=|z_i-z_i'|_{\mathbb T}$.
We say that the arc is \emph{cut} if $U_i$ lies on it, including
either endpoint.
The following fact collects two geometric properties of the random cut.

\begin{fact}
\label{lem:random-cut}
The following statements hold.
\begin{enumerate}[label=\textup{(\roman*)},ref=\textup{(\roman*)}]
\item \label{item:fact-cut-item-v} For fixed $z,z'$, the probability that at least one of the chosen
  shortest coordinate arcs from $z$ to $z'$ is cut satisfies
  \[
    \Pp_U\{\text{at least one chosen arc is cut}\}
    \le\sum_{i=1}^d\Delta_i(z,z')
    =O\!\left(d_{\mathbb T}(z,z')\right).
  \]
\item \label{item:fact-cut-item-vi} For nonempty and finite $S \subseteq \mathbb T^d$, let $b$ be a nearest neighbor of $x$ in $S$ under the torus metric.
  If none of the shortest coordinate arcs from $x$ to $b$ is cut by $U$, then $b$ is also a nearest neighbor of $x$ in $S$ under the
  cut metric $d_U$.
\end{enumerate}
\end{fact}

Both statements in Fact~\ref{lem:random-cut} follow from the definition of the torus distance and the uniform choice of cuts.
In particular, Property~\ref{item:fact-cut-item-vi} will let us control
disagreements between the Greedy processes on the torus and in the cube.
The proof of Fact~\ref{lem:random-cut} is deferred to
Appendix~\ref{app:proof-random-cut}.

\subsection{Randomly Shifted Grids}

When proving the cost lower bound for Greedy, we will impose randomly shifted grids.
To construct the grid, fix a desired cell size $0<h\le1$ and use side
length $s:=\lceil1/h\rceil^{-1}\in[h/2,h]$.  
Independently of the inputs, sample the \emph{grid shift} $V=(V_1,\ldots,V_d)\sim\operatorname{Unif}([0,s)^d)$.
The Euclidean grid consists of the half-open cells
\[
  P_j(V):=V+sj+[0,s)^d,\qquad j\in\mathbb Z^d.
\]
On the torus, the cells are
\[
  P_j^{\mathbb T}(V):=\{z\bmod1:z\in P_j(V)\},
  \qquad j\in\{0,\ldots,s^{-1}-1\}^d.
\]
Since $s^{-1}$ is an integer, these $s^{-d}$ cells partition
$\mathbb T^d$, and each has volume $s^d$.
A torus cell is counted as one cell even if its representation in
$[0,1)^d$ consists of several pieces.

The following standard separation estimate is part of the random-dissection
method used in geometric approximation algorithms.\footnote{See
\cite{DBLP:journals/jacm/Arora98} and the explicit grid bound in
\cite{DBLP:journals/siamcomp/AigerKS14}.}
It bounds the chance of separating two fixed points by a constant times
their distance divided by the cell size.

\begin{lemma}
\label{lem:random-grid-separation}
For the random grid just defined, let $x,y$ be two fixed points in
$\mathbb T^d$ (resp. $\mathbb R^d$), and let $r$ be their distance
under the torus metric (resp. the Euclidean metric).  Then
\begin{equation}\label{eq:random-grid-separation}
\Pp_V\{x\text{ and }y\text{ lie in different cells}\}
  =O\!\left(\min\left\{1,\frac rh\right\}\right).
\end{equation}
\end{lemma}

Appendix~\ref{app:proof-random-grid-separation} records the source of
Lemma~\ref{lem:random-grid-separation} and its application to the torus grid.

\subsection{Periodic Voronoi Geometry}
\label{sec:periodic-voronoi}

In this subsection, we establish several properties of periodic Voronoi
diagrams.

For a finite nonempty $S\subset\mathbb T^d$, define the distance to $S$
and the functional $Q$\footnote{The functional $Q$ is the periodic version,
with uniform density, of the quadratic quantization energy studied in
the theory of centroidal Voronoi
tessellations~\cite{du1999centroidal}.  It measures the mean
squared error when each location is represented by its nearest server.
Here the configuration $S$ is arbitrary, and no assumption that servers
are the centroids of their cells is imposed.} by
\[
  d_{\mathbb T}(x,S):=\min_{s\in S}d_{\mathbb T}(x,s),
  \qquad
  Q(S):=\int_{\mathbb T^d}d_{\mathbb T}(x,S)^2\,\mathrm{d}x,
\]
and, for $s\in S$ with $|S|\ge2$, define the increase in $Q$ caused by
deleting server $s$ from $S$ as
\[
  \Delta_sQ(S):=Q(S\setminus\{s\})-Q(S).
\]
Fix such a state $S$ and server $s$, using representatives in $[0,1)^d$.
We now use the periodic-lift viewpoint in $\mathbb R^d$ and
translate the entire periodic configuration by $-s$.
The server copies are then $t-s+z$, for $t\in S$ and
$z\in\mathbb Z^d$, with the chosen copy of $s$ at the origin.
Its Euclidean Voronoi cell consists of points whose distance to
the origin is no greater than their distance to any other copy,
including the other copies of $s$.
The cell and its volume are respectively
\[
  \mathcal V_s:=\left\{x\in\mathbb R^d:
    \|x\|_2\le\|x-(t-s+z)\|_2
    \text{ for all }t\in S,\ z\in\mathbb Z^d\right\},
  \qquad v_s:=|\mathcal V_s|.
\]
The cell $\mathcal V_s$ is a bounded convex polytope contained in
$[-1/2,1/2]^d$.
Under $x\mapsto s+x\pmod1$, it represents the torus
Voronoi cell of $s$.  In particular,
$v_s$ equals the probability that Greedy matches a uniform request to $s$.

Note that deleting server $s$ also removes all of its periodic copies from the
lifted configuration.
For $x\in\mathcal V_s$, let $T_s(x)$ be the location of a nearest
periodic copy of a surviving server after $s$ is deleted.  Equivalently,
write $T_s(x)=t-s+z$, where the surviving server
$t\in S\setminus\{s\}$ and the integer translation $z\in\mathbb Z^d$
are chosen to attain the minimum in
\[
  \|x-T_s(x)\|_2
  =\min_{\substack{t\in S\setminus\{s\}\\z\in\mathbb Z^d}}
       \|x-(t-s+z)\|_2.
\]
Ties can be resolved by any fixed
measurable rule.

The following lemma, whose proof is deferred to
Appendix~\ref{app:proof-periodic-deletion}, shows that $T_s(x)$ integrates to zero over
the original cell $\mathcal V_s$.
It also bounds the integral of its squared length in terms of
$\Delta_sQ(S)$, the increase in mean squared nearest-server distance
caused by deleting $s$.

\begin{lemma}
\label{lem:periodic-voronoi-deletion}
For every finite set $S\subset\mathbb T^d$ with $|S|\ge2$ and every $s\in S$,
\[
  \int_{\mathcal V_s}T_s(x)\,\mathrm{d}x=0,
  \qquad
  \int_{\mathcal V_s}\|T_s(x)\|_2^2\,\mathrm{d}x\le(d+1)\Delta_sQ(S).
\]
\end{lemma}

\begin{remark}
The first identity in Lemma~\ref{lem:periodic-voronoi-deletion} is a
periodic counterpart of Sibson's classical vector identity~\cite{sibson1980vector}.
In the classical identity, only one point is deleted.
In the above periodic setting, deleting $s$ removes every copy of $s$.
Hence, we cannot directly apply the classical identity here since it would
leave all other copies of $s$ available, so points could be reassigned
to those copies.
\end{remark}

\section{Properties of the Greedy Process}
\label{sec:greedy-properties}

In this section, we establish several properties of the Greedy process.
We first use stability under input changes to bound server-count
variances and matching costs on the torus.
We then use these cost bounds to compare the torus and cube
Greedy processes.

\subsection{Stability and Server-Count Variance}
\label{sec:stability-variance}

In this subsection, we present a stability property of Greedy and bound the overall variance of
remaining-server counts over an arbitrary finite measurable partition.
For two finite point configurations $\Sigma$ and $\Sigma'$ of the same size, we say that they \emph{differ by at most one point replacement} if $|\Sigma\setminus\Sigma'|=|\Sigma'\setminus\Sigma|\le1$.
We first recall the following
standard coupling property established
in~\cite[proof of Lemma~13(1)]{DBLP:conf/sigecom/BalkanskiFP23}.

\begin{lemma}[\cite{DBLP:conf/sigecom/BalkanskiFP23}]\label{lem:greedy-stability}
Consider two Greedy processes using the same fixed metric and
request sequence, with initial server configurations that
differ by at most one point replacement.
If all nearest-server choices in both processes are unique, their
remaining configurations differ by at most one point replacement
after every request.
\end{lemma}

Fix $1\le m\le n$, and let $(\mathsf X,\delta)$ be a metric
space equipped with its Borel sigma-field.
Let the servers $Y_1,\ldots,Y_n$ and requests
$X_1,\ldots,X_{n-m}$ be independent $\mathsf X$-valued random points.
Let $\Sigma_m$ be the remaining-server configuration after
Greedy serves these requests.
For a measurable set $P$ and a finite point configuration $\Sigma$, define $N_P(\Sigma):=|\Sigma\cap P|$.
The next lemma bounds the total variance of server counts over an arbitrary space
partition, without a factor for the number of cells.

\begin{lemma}
\label{prop:aggregate-variance}
{Assume that the server locations are pairwise distinct almost surely and
that, for each $1\le i\le n-m$, the distances
$\delta(X_i,Y_1),\ldots,\delta(X_i,Y_n)$ are pairwise distinct almost surely.
Then, for every finite Borel partition $\mathcal P$ of
$\mathsf X$,}
\begin{equation}\label{eq:aggregate-variance}
  \sum_{P\in\mathcal P}\Var\!\bigl(N_P(\Sigma_m)\bigr)\le2n-m.
\end{equation}
\end{lemma}

\begin{proof}
\phantomsection\label{sec:proof-aggregate-variance}
Index the $2n-m$ independent inputs by writing
\[
  Z=(Z_1,\ldots,Z_{2n-m})
   :=(Y_1,\ldots,Y_n,X_1,\ldots,X_{n-m}).
\]
To apply \eqref{eq:efron-stein-resampling}, for each $1\le j\le2n-m$, let $Z^{(j)}$ be obtained from $Z$ by replacing
only $Z_j$ with an independent copy.

For an input vector $z$, with coordinates ordered as in $Z$, let
$\Sigma_m(z)$ denote the remaining-server configuration after Greedy
serves its $n-m$ requests.
Let $\Sigma_m:=\Sigma_m(Z)$ and
$\Sigma_m^{(j)}:=\Sigma_m(Z^{(j)})$.
If $j\le n$, Lemma~\ref{lem:greedy-stability} gives at most one point
replacement between $\Sigma_m$ and $\Sigma_m^{(j)}$.
If $j>n$, the processes agree until the changed request; immediately
afterward their configurations differ by at most one point replacement,
which Lemma~\ref{lem:greedy-stability} preserves through subsequent requests.
Hence, at most two cells have their remaining-server counts changed, each by one in absolute value.
Thus, for every $j$, almost surely,
\[
  \sum_{P\in\mathcal P}
    \bigl(N_P(\Sigma_m)-N_P(\Sigma_m^{(j)})\bigr)^2\le2.
\]
Apply \eqref{eq:efron-stein-resampling} to each function
$f_P(Z)=N_P(\Sigma_m(Z))$, which is bounded by $m$, and sum over the
cells to obtain
\[
  \sum_{P\in\mathcal P}\Var\!\bigl(N_P(\Sigma_m)\bigr)
  \le\frac12\sum_{j=1}^{2n-m}
          \E\sum_{P\in\mathcal P}
          \bigl(N_P(\Sigma_m)-N_P(\Sigma_m^{(j)})\bigr)^2
  \le2n-m,
\]
as desired.
\end{proof}

\subsection{Matching Costs on the Torus}
\label{sec:torus-matching-distances}

We now combine the server-count variance bound from
Lemma~\ref{prop:aggregate-variance} with the translation-symmetry property of torus to
bound Greedy's matching costs on the torus.

Recall the torus configuration $B_m$ from
Fact~\ref{fact:torus-translation-stationarity}.
For $1\le m\le n$, let $D_m^{\mathbb T} :=d_{\mathbb T}(X^{(m)},B_m)$ be Greedy's matching cost
for request $X^{(m)}$ under the torus metric.
The following lemma gives a tail bound for the matching cost at each rank.

\begin{lemma}
\label{prop:torus-edge}
For every $1\le m\le n$ and $0<r\le1$,
\begin{equation}\label{eq:torus-tail}
\Pp\{D_m^{\mathbb T}>\sqrt d\,r\}
  =O\!\left(\min\left\{1,\frac{n}{m^2r^d}\right\}\right).
\end{equation}
\end{lemma}

\begin{proof}
Fix $0<r\le1$,
and let $\mathcal P$ be the partition of the torus into half-open cubes of side $s=\lceil1/r\rceil^{-1}\in[r/2,r]$.
The fresh request lies in each cell $P\in\mathcal P$ with probability
$|P|$.  If that cell contains an unmatched server, the matching cost is
at most $\sqrt d\,s$.  Independence of the fresh request and $B_m$ gives
\[
  \Pp\{D_m^{\mathbb T}>\sqrt d\,r\}
  \le\sum_{P\in\mathcal P}|P|\Pp\{N_P(B_m)=0\}.
\]
By Chebyshev's inequality and
Fact~\ref{fact:torus-translation-stationarity}, we have
\[
  \Pp\{N_P(B_m)=0\}
  \le\frac{\Var(N_P(B_m))}{m^2|P|^2}.
\]
{Since every cell has volume $s^d$, summing the preceding bound and
applying Lemma~\ref{prop:aggregate-variance} with $\Sigma_m=B_m$
gives
\begin{align*}
  \Pp\{D_m^{\mathbb T}>\sqrt d\,r\}
  \le \frac{1}{m^2s^d}
       \sum_{P\in\mathcal P}\Var\!\bigl(N_P(B_m)\bigr)
  \le \frac{2n-m}{m^2s^d}
   =O\!\left(\frac{n}{m^2r^d}\right),
\end{align*}
where the equality holds since $s\ge r/2$.  Combining this estimate with
the trivial probability bound by one proves \eqref{eq:torus-tail}.}
\end{proof}

Define the rank-dependent length scale
\begin{equation}\label{eq:intro-rank-scale}
  \ell_m
  :=\min\left\{1,\left(\frac{n}{m^2}\right)^{1/d}\right\},
\end{equation}
and define the total-cost scale
\begin{equation}\label{eq:rank-cutoff}
  M_n:=
  \begin{cases}
    \sqrt n, & d=1,\\
    \sqrt n\,\log n, & d=2,\\
    n^{1-1/d}, & d\ge3.
  \end{cases}
\end{equation}
The following corollary gives bounds on expected matching costs.

\begin{corollary}\label{cor:torus-mean-costs}
The expected total matching cost satisfies
\begin{equation}\label{eq:torus-total-one}
\sum_{m=1}^n\E D_m^{\mathbb T}=O(M_n).
\end{equation}
For $d\ge2$ and every $1\le m\le n$, the expected matching cost at rank $m$ satisfies
\begin{equation}\label{eq:torus-edge}
\E D_m^{\mathbb T}=O(\ell_m).
\end{equation}
\end{corollary}

The proof of \Cref{cor:torus-mean-costs}, which is given in Appendix~\ref{app:proof-torus-mean-costs}, proceeds by integrating the tail estimate in
Lemma~\ref{prop:torus-edge} and summing over ranks.

\subsection{Comparing Greedy on Torus and Cube}
\label{sec:random-cut-discrepancy}

In this subsection, we provide tools to transfer estimates for torus Greedy to Greedy under Euclidean
distances in the cube, where translation symmetry no longer holds.

Draw iid uniform servers and requests on $\Td$, together with the
independent random cut $U$ from \Cref{sec:random-cuts-unwrapping}.
We aim to compare Greedy under the cut metric $d_U$ and the torus metric
$d_{\mathbb T}$, using the same server locations and request order.
Let $A_m$ be the servers remaining after the first $n-m$ requests when
Greedy uses the cut metric $d_U$, and let $B_m$ be the servers remaining
when Greedy uses the torus metric $d_{\mathbb T}$.
Recall that $S_m$ is the set of $m$ unmatched servers in the
Greedy process on the original cube inputs with the Euclidean metric.

We apply $\kappa_U$ to every server and request location in the Greedy
process under $d_U$, replacing each match $(x,s)$ by
$(\kappa_U(x),\kappa_U(s))$.
At rank $m$, the transformed remaining-server configuration is
$\kappa_U(A_m)$.
The following lemma identifies the law of this transformed process.

\begin{lemma}\label{lem:unwrapped-greedy-law}
\label{item:fact-cut-item-ii}
Condition on $U$.
Then, the transformed
inputs, remaining-server configurations at all ranks, and matching costs
jointly have the same law as the Euclidean Greedy process in the cube.
In particular,
\[
  \bigl(\kappa_U(A_m)\bigr)_{0 \leq m \leq n}
  \overset{\mathrm d}=\bigl(S_m\bigr)_{0 \leq m \leq n}.
\]
\end{lemma}

The proof of \Cref{lem:unwrapped-greedy-law} is given in Appendix~\ref{app:proof-unwrapped-greedy-law}.
We now compare $A_m$ and $B_m$.
Note that the torus Greedy process $(B_m)_{m=0}^n$ is independent of $U$.
Define the discrepancy
\[
  K_m:=|A_m\setminus B_m|=|B_m\setminus A_m|,
  \qquad 0\le m\le n.
\]
We show in the next lemma that $\Expect K_m$ can be bounded by the total-cost scale
$M_n$ from \eqref{eq:rank-cutoff}.

\begin{lemma}\label{lem:discrepancy}
For every $0\le m\le n$, the coupled Greedy processes under
$d_U$ and $d_{\mathbb T}$ satisfy
\begin{equation}\label{eq:K-expected}
\E K_m=O\!\left(\min\{m,M_n\}\right).
\end{equation}
\end{lemma}

\begin{proof}
Intuitively, a random cut is unlikely to
change a short match, and each change can add at most one server to the
discrepancy.  Combining this observation with Corollary~\ref{cor:torus-mean-costs}
gives the desired bound.

We first bound $\E K_m$ by the expected total matching
cost of the torus Greedy process.
To track the choices of the cut and torus Greedy processes,
we use the following notation.
For a metric $\delta$, let $\NN_S^\delta(x)$ denote the unique nearest
point of a nonempty finite point set $S$ to a point $x$ under $\delta$.
For $1\le m\le n$, define
\[
  a_m:=\NN_{A_m}^{d_U}(X^{(m)}),
  \qquad
  b_m:=\NN_{B_m}^{d_{\mathbb T}}(X^{(m)}).
\]
Thus $a_m$ and $b_m$ are the servers actually deleted by the cut and torus
processes, respectively.  To isolate changes caused by the metric,
compare the nearest server in the same set $B_m$ under the two metrics.
Define the event that these choices differ and its probability by
\[
  E_m:=\{b_m\ne\NN_{B_m}^{d_U}(X^{(m)})\},
  \qquad
  p_m:=\Pp(E_m).
\]
By Fact~\ref{lem:random-cut}\ref{item:fact-cut-item-vi}, $E_m$
can occur only if, for at least one coordinate $i$, the cut point $U_i$
lies on the chosen shortest arc joining the $i$th coordinates of
$X^{(m)}$ and $b_m$.
For $1\le m\le n$, we first argue that
\begin{equation}\label{eq:K-recursion}
  K_{m-1}\le K_m+\mathbf1_{E_m}.
\end{equation}
To see this, notice that the discrepancy can increase only if the cut and torus
Greedy processes match a request to distinct servers in $A_m \cap B_m$, and, in this case, the discrepancy
increases by precisely one.
On $E_m^c$, the definition of $E_m$ says that $b_m$ is also the unique
nearest server in $B_m$ under $d_U$.  If $a_m$ and $b_m$ both belong to
$A_m\cap B_m$, each is the unique nearest server to $X^{(m)}$ within
this intersection under $d_U$, so $a_m=b_m$.
Thus the discrepancy cannot increase on $E_m^c$, proving
\eqref{eq:K-recursion}.  Iterating from $K_n=0$ and taking expectations
gives
\begin{equation}\label{eq:K-bound}
  \E K_m\le\sum_{k=m+1}^np_k.
\end{equation}

Notice that the torus configuration $B_m$ and the request $X^{(m)}$ are jointly
independent of $U$.  Conditional on $X^{(m)}$ and $B_m$, the server
$b_m$ and the chosen shortest coordinate arcs from $b_m$ to $X^{(m)}$ are fixed, while $U$ remains uniform.
Fact~\ref{lem:random-cut}\ref{item:fact-cut-item-v} therefore bounds
the conditional probability of cutting at least one of these arcs,
and hence of $E_m$, by $\sum_{i=1}^d\Delta_i(X^{(m)},b_m)$.
Taking expectation over the random request $X^{(m)}$ and configuration
$B_m$ gives
\begin{equation}\label{eq:seam-probability}
  p_m
  \le\E\sum_{i=1}^d
       \Delta_i\bigl(X^{(m)},b_m\bigr)
  =O\!\left(\E D_m^{\mathbb T}\right).
\end{equation}
Combining \eqref{eq:K-bound} with \eqref{eq:seam-probability} gives
\begin{equation*}
  \E K_m
  \le\sum_{k=m+1}^n p_k
  =O\!\left(\sum_{k=m+1}^n\E D_k^{\mathbb T}\right).
\end{equation*}
The total-cost bound in Corollary~\ref{cor:torus-mean-costs} bounds the last
expression by $O(M_n)$.  Also, $K_m\le m$ because each configuration
contains $m$ servers.  This proves \eqref{eq:K-expected}.
\end{proof}

\section{Competitive Ratio Upper Bound}\label{sec:upper}

In this section, we prove \Cref{thm:costs}.
Define $D_m:=\dist(X^{(m)},S_m)$ as Euclidean Greedy's matching cost in the cube for request $X^{(m)}$.
The set $S_m$ depends only on the servers and the earlier requests
$X^{(n)},\ldots,X^{(m+1)}$.  Hence, $X^{(m)}$ is uniform and independent of
$S_m$, and, for every input sequence,
\begin{equation}\label{eq:rank-sum}
  G_n=\sum_{m=1}^nD_m.
\end{equation}

We use the total-cost scale $M_n$ from \eqref{eq:rank-cutoff} as a
rank cutoff.  The next lemma bounds the expected matching cost
in the cube at each rank.

\begin{lemma}\label{prop:rankwise}
For $d=1$ and every $1\le m\le n$,
\begin{equation}\label{eq:rankwise-mean-one}
\E D_m=
  \begin{cases}
    O(1), & 1\le m\le\sqrt n,\\[2mm]
    \displaystyle O\!\left(\frac{n}{m^2}\left(1+\log\frac{m^2}{n}\right)\right),
      & \sqrt n<m\le n.
  \end{cases}
\end{equation}
For $d\ge2$ and every $1\le m\le n$,
\begin{equation}\label{eq:rankwise-mean}
\E D_m=
  \begin{cases}
    O(1), & 1\le m\le M_n,\\[2mm]
    \displaystyle O\!\left(\frac{n^{1/d}}{m^{2/d}}
      +\left(\frac{M_n}m\right)^{1+1/d}\right),
      & M_n<m\le n.
  \end{cases}
\end{equation}
\end{lemma}

The proof of Lemma~\ref{prop:rankwise} is deferred to
\Cref{sec:proof-rankwise}.
We are ready to prove Theorem~\ref{thm:costs}.

\begin{proof}[Proof of Theorem~\ref{thm:costs}]
For $d=1$, \eqref{eq:rank-sum} and Lemma~\ref{prop:rankwise} give
\[
  \E G_n
  =O\!\left(\sqrt n+n\,\int_{\sqrt n}^{\infty}
       \frac{1+\log(x^2/n)}{x^2}\,\mathrm{d}x\right)
  =O(\sqrt n).
\]
For $d\ge2$, the low ranks contribute $O(M_n)$, and
\[
  \E G_n=O\!\left(
    M_n+n^{1/d}\,\int_{M_n}^{n}x^{-2/d}\,\mathrm{d}x
    +M_n^{1+1/d}\,\int_{M_n}^{\infty}x^{-1-1/d}\,\mathrm{d}x\right).
\]
The last term is $O(M_n)$.  The middle term
is $O(\sqrt n\log n)$ when $d=2$ and $O(n^{1-1/d})$ when $d\ge3$.
Thus, in every dimension,
\begin{equation}\label{eq:master-upper}
\E G_n=O(M_n).
\end{equation}
On the line, sorted matching is optimal, and the standard one-dimensional
random matching scale is $\Theta(\sqrt n)$.  Together with the
classical higher-dimensional estimates~\cite{ajtai1984optimal,talagrand1992matching}, the offline scales are
\[
  \E\OPT_n=
  \begin{cases}
    \Theta(\sqrt n), & d=1,\\
    \Theta(\sqrt{n\log n}), & d=2,\\
    \Theta(n^{1-1/d}), & d\ge3.
  \end{cases}
\]
Dividing \eqref{eq:master-upper} by the offline matching scales gives the
competitive-ratio upper bounds.
\end{proof}

\subsection{Proof of Lemma~\ref{prop:rankwise}}
\label{sec:proof-rankwise}

Since $D_m\le\sqrt d$, we have $\E D_m=O(1)$.  This proves
\eqref{eq:rankwise-mean-one} for $m\le\sqrt n$ when $d=1$, and
\eqref{eq:rankwise-mean} for $m\le M_n$ when $d\ge2$.
For the remainder of the proof, assume $m>M_n$ and recall the length
scale $\ell_m$ from \eqref{eq:intro-rank-scale}.
Since $M_n\ge\sqrt n$ in every dimension, we have
\[
  \ell_m=\frac{n^{1/d}}{m^{2/d}},
  \qquad
  \ell_m^d=\frac{n}{m^2}.
\]

For a measurable region $P\subseteq\Qd$, let $\rho_m(P) := \E N_P(S_m) / m$ be its
expected remaining-server count divided by $m$.
Thus $\rho_m$ is a probability measure and $\rho_m(\Qd)=1$.
Fix $\ell_m\le r\le1$ and partition the cube into half-open cells
of side length $s=\lceil1/r\rceil^{-1}\in[r/2,r]$.  For each cell $P$, write
$|P|=s^d$ for its volume.
Call $P$ \emph{bad} if $\rho_m(P)<|P|/2$ and \emph{good} otherwise.
Note that goodness of a cell is a deterministic property.
Define
$\beta_m(r):=\sum_{P\ \mathrm{bad}}|P|$ as the total volume of bad cells.

Since each cell has diameter at most $\sqrt d\,r$, a match with a cost larger than
this can occur only when the request lies in a bad cell or in an empty
good cell.  The fresh request
is uniform and independent of $S_m$, which gives
\begin{equation}\label{eq:cube-tail-decomposition}
  \Pp\{D_m>\sqrt d\,r\}
  \le \beta_m(r)
     +\sum_{P\ \mathrm{good}}|P|\Pp\{N_P(S_m)=0\}.
\end{equation}
For a good cell,
\begin{equation*}
  \E N_P(S_m)=m\rho_m(P)\ge\frac{m|P|}{2}.
\end{equation*}
Chebyshev's inequality bounds its probability of being empty by
$4\Var(N_P(S_m))/(m^2|P|^2)$.  Sum this bound in
\eqref{eq:cube-tail-decomposition}, using $|P|=s^d\ge(r/2)^d$ and
the aggregate variance bound of Lemma~\ref{prop:aggregate-variance}:
\begin{align}
  \Pp\{D_m>\sqrt d\,r\}
  \le \beta_m(r)
       +\frac{4}{m^2s^d}
          \sum_{P\ \mathrm{good}}\Var\!\bigl(N_P(S_m)\bigr)
  \le \beta_m(r)+O\!\left(\frac{n}{m^2r^d}\right). \label{eq:cube-tail-intermediate}
\end{align}
It remains to bound the total volume $\beta_m(r)$ of the bad cells.
For Greedy using Euclidean distances,
the expected number of remaining servers in a region $P$ is $m\rho_m(P)$.
For Greedy using the torus metric,
Fact~\ref{fact:torus-translation-stationarity} gives its expected remaining-server count
$m|P|$.
We measure the difference between these expectations, divided by $m$, by
defining
\[
  H_m:=\sup_{P\subseteq\Qd\ \mathrm{measurable}}|\rho_m(P)-|P||.
\]
This is the largest normalized expected-count error over all regions,
or equivalently the total variation distance of the measure $\rho_m$ from uniform measure.
The following lemma upper bounds $H_m$.

\begin{lemma}\label{prop:mean-bias}
For every $1\le m\le n$,
\begin{equation}\label{eq:mean-bias-transfer}
H_m=O\!\left(\min\left\{1,\frac{M_n}m\right\}\right).
\end{equation}
\end{lemma}

The proof of Lemma~\ref{prop:mean-bias} is deferred to
the end of this subsection.  Applying the definition of $H_m$ to the union of
the bad cells gives
\[
  H_m\ge\sum_{P\ \mathrm{bad}}\bigl(|P|-\rho_m(P)\bigr)
  \ge\frac{\beta_m(r)}2,
\]
where the last inequality holds by the definition of bad cells.
By Lemma~\ref{prop:mean-bias}, there is a constant $C_0$ such that
$H_m\le C_0M_n/m$, so the preceding inequality gives
$\beta_m(r)\le2C_0M_n/m$.  Moreover, if a bad cell $P$ exists, then
its volume $|P|=s^d\ge(r/2)^d$ gives
\[
  C_0\frac{M_n}{m}\ge H_m\ge |P|-\rho_m(P)
  >\frac{|P|}{2}\ge2^{-d-1}r^d.
\]
Thus a bad cell can exist only when $r^d<2^{d+1}C_0M_n/m$.
Taking $C:=2^{d+1}C_0$ therefore yields
\begin{equation}\label{eq:beta-bounds}
  \beta_m(r)=O\!\left(\frac{M_n}m\right),
  \qquad
  \beta_m(r)=0
  \quad\text{if}\quad
  r^d\ge C\frac{M_n}m.
\end{equation}
Since $m>M_n$, we have $\ell_m=n^{1/d}m^{-2/d}$.
Combining \eqref{eq:cube-tail-intermediate} and
\eqref{eq:beta-bounds} therefore gives
\begin{equation}\label{eq:cube-tail}
\Pp\{D_m>\sqrt d\,r\}
  =O\!\left(\frac{M_n}m
      \mathbf1\!\left\{r^d<C\frac{M_n}m\right\}
      +\left(\frac{\ell_m}r\right)^d\right).
\end{equation}

Finally, we integrate the above tail probability.  Since $0\le D_m\le\sqrt d$,
the change of variables $t=\sqrt d\,r$ gives
\[
  \E D_m=\sqrt d\int_0^1
       \Pp\{D_m>\sqrt d\,r\}\,\mathrm{d}r.
\]
{We use the trivial bound by one on $0\le r<\ell_m$ and
\eqref{eq:cube-tail} on the remaining interval.}  Suppose first that
$d\ge2$.  Then
\begin{align*}
  \E D_m
  =O\!\left(\ell_m+\frac{M_n}m
        \int_0^{\min\{1,(CM_n/m)^{1/d}\}}\mathrm{d}r
     +\ell_m^d\int_{\ell_m}^1r^{-d}\,\mathrm{d}r\right)
  =O\!\left(\frac{n^{1/d}}{m^{2/d}}
       +\left(\frac{M_n}m\right)^{1+1/d}\right).
\end{align*}
This proves \eqref{eq:rankwise-mean} for $m>M_n$.  For $d=1$,
we have $M_n=\sqrt n$ and $\ell_m=n/m^2$ when $m>M_n$.
Using \eqref{eq:cube-tail} in the same integral formula gives
\begin{align*}
  \E D_m
  =O\!\left(\ell_m+\left(\frac{M_n}m\right)^2
     +\ell_m\int_{\ell_m}^1\frac{\mathrm{d}r}{r}\right)
  =O\!\left(\frac{n}{m^2}\left(1+\log\frac{m^2}{n}\right)\right).
\end{align*}
This proves
\eqref{eq:rankwise-mean-one} for $m>\sqrt n$.

\begin{proof}[Proof of Lemma~\ref{prop:mean-bias}]
Fix $1\le m\le n$.
To bound $H_m$, we must control $|\rho_m(P)-|P||$ uniformly over $P$.
Use the coupling from \Cref{sec:random-cut-discrepancy}, in
which $A_m$ and $B_m$ are the remaining servers under the cut metric
$d_U$ and the torus metric $d_{\mathbb T}$, respectively.
Lemma~\ref{lem:unwrapped-greedy-law} gives
\[
  \E|\kappa_U(A_m)\cap P|=m\rho_m(P).
\]
Greedy under $d_{\mathbb T}$ does not use $U$, so $B_m$ is independent
of $U$.  Conditional on $B_m$, each point $b\in B_m$ has a uniform
unwrapped location $\kappa_U(b)$.  Summing over its $m$ points gives
\[
  \E|\kappa_U(B_m)\cap P|=m|P|.
\]

Recall that $K_m=|A_m\setminus B_m|=|B_m\setminus A_m|$.
For every server $b\in A_m\cap B_m$, its image $\kappa_U(b)$ belongs
to $P$ in both configurations or in neither.  Thus these common servers
contribute equally to $|\kappa_U(A_m)\cap P|$ and
$|\kappa_U(B_m)\cap P|$.  The sets $A_m\setminus B_m$ and
$B_m\setminus A_m$ each contain $K_m$ servers, so the difference between
these two cardinalities is at most $K_m$ in absolute value:
\begin{equation}\label{eq:empirical-TV-coupling}
  \left||\kappa_U(A_m)\cap P|-|\kappa_U(B_m)\cap P|\right|\le K_m.
\end{equation}
The two mean identities above, followed by the triangle inequality for
expectations and \eqref{eq:empirical-TV-coupling}, give
\begin{align*}
  m|\rho_m(P)-|P||
  &=\left|\E|\kappa_U(A_m)\cap P|-\E|\kappa_U(B_m)\cap P|\right| \\
  &\le\E\Bigl|\,|\kappa_U(A_m)\cap P|-|\kappa_U(B_m)\cap P|\,\Bigr|
  \le\E K_m.
\end{align*}
Divide by $m$ and take the supremum over $P$ to obtain
\begin{equation*}
  H_m\le\frac{\E K_m}{m}.
\end{equation*}
Finally, Lemma~\ref{lem:discrepancy} gives
$\E K_m=O(\min\{m,M_n\})$, which proves
\eqref{eq:mean-bias-transfer}.
\end{proof}
\section{Competitive Ratio Lower Bound for \texorpdfstring{$d=2$}{d=2}}\label{sec:two-d-lower}

In this section, we prove \Cref{thm:lb-2d}.
Throughout the proof, $n$ is sufficiently large after the universal
constants have been fixed.
We use the cube configuration $S_m$, the coupled torus configurations
$A_m,B_m$, and their discrepancy $K_m$ from
\Cref{sec:random-cut-discrepancy}.
For the square matching costs $D_k$, define
\[
  G_{[a,b]}:=\sum_{k=\lceil a\rceil}^{\lfloor b\rfloor}D_k.
\]
For integers $a<b$, the transition from rank $b$ to rank $a$ incurs cost
$G_{[a+1,b]}$.

\paragraph{The rank-interval estimate that suffices.}
Fix a sufficiently small universal constant $c_0\in(0,1/4]$, chosen
to satisfy the restrictions in the lemmas below.  We call a rank $m$ \emph{admissible} if $n^{3/4}\le m\le c_0n$, and, for a sufficiently large constant $C_1$, define
\refstepcounter{equation}\label{eq:intro-localization-scale}\label{eq:adaptive-scale}
\begin{equation*}
  h_m:=C_1\frac{\sqrt n}{m},
  \qquad T_m:=\frac{\sqrt n}{h_m}=\frac{m}{C_1}.
  \tag{\theequation}
\end{equation*}
We will prove that
for some constant $c>0$, every admissible rank $m$ satisfies
\begin{equation}\label{eq:square-block-charge}
\E G_{[cT_m,m]}=\Omega(\sqrt n).
\end{equation}
To see why \eqref{eq:square-block-charge} suffices, decrease $c$ if
necessary so that $\alpha:=c/C_1\in(0,1)$.
Start with $m_0=\lfloor c_0n\rfloor$ and choose
\[
  m_{j+1}=\left\lfloor\frac{\alpha m_j}{2}\right\rfloor,
\]
retaining these ranks as long as $m_j\ge n^{3/4}$.
Since $cT_{m_j}=\alpha m_j$ and $m_{j+1}<\alpha m_j$,
the intervals $[cT_{m_j},m_j]$ are disjoint.  Their endpoints decrease
by a fixed factor, so there are $\Theta(\log n)$ such intervals.
Summing \eqref{eq:square-block-charge} gives $\E G_n=\Omega(\sqrt n\log n)$.
Dividing by the classical offline scale
$\E\OPT_n=\Theta(\sqrt{n\log n})$~\cite{ajtai1984optimal,bobkov2021simple}
proves \Cref{thm:lb-2d}.
It remains to establish \eqref{eq:square-block-charge}.

\paragraph{Reduction to a server deficit in the square.}
Fix an admissible rank $m$.  We will deduce
\eqref{eq:square-block-charge} from a lower bound on the total server
deficit in square grid cells of side roughly $h_m$.
For a set $\Sigma$ of $m$ remaining servers, define the \emph{server
deficit} in a cell $P$ as
\[
  (m|P|-N_P(\Sigma))_+,
  \qquad z_+:=\max\{z,0\}.
\]
Here, $m|P|$ is the expected number of the $m$ future requests that lie
in $P$.  The \emph{total deficit} is the sum of the server deficits over all cells.
Intuitively, any requests in excess of the servers available in their cell must
be matched outside it, yielding a lower bound on total matching cost.

Consider the original square Greedy process $S_k$.  Set
$s=\lceil1/h_m\rceil^{-1}\in[h_m/2,h_m]$ and choose a grid shift
$V\sim\operatorname{Unif}([0,s)^2)$ independently of the inputs.
Let $\mathcal C_V$ consist of the squares $V+sj+[0,s)^2$,
$j\in\mathbb Z^2$, that lie entirely in $[0,1]^2$.
We only use the grids for analysis, and Greedy's
matches do not depend on $V$.
We show in the following lemma that a large total deficit would imply the desired cost lower bound.

\begin{samepage}
\begin{lemma}\label{prop:delayed-service}
Consider Euclidean Greedy on the square and the independent grid
$\mathcal C_V$ defined above.
For every fixed constant $d_0\in(0,1]$, there exists a constant $c>0$,
depending only on $d_0$, with the following property.
For every admissible rank $m$, if
\begin{equation}\label{eq:service-hypothesis}
  \E\sum_{P\in\mathcal C_V}(m|P|-N_P(S_m))_+\ge d_0T_m,
\end{equation}
then
\begin{equation}\label{eq:service-conclusion}
\E G_{[cT_m,m]}=\Omega(\sqrt n).
\end{equation}
Both expectations are over the input locations and the grid shift $V$.
\end{lemma}
\end{samepage}

The proof of Lemma~\ref{prop:delayed-service} is deferred to \Cref{sec:proof-delayed-service}.
Next, we establish the deficit hypothesis \eqref{eq:service-hypothesis} on the square
from a corresponding deficit bound for the torus configuration.

\paragraph{From torus server deficit to square server deficit.}
Fix an admissible rank $m$, and recall that $s=\lceil1/h_m\rceil^{-1}\in[h_m/2,h_m]$.
For a uniform shift $V\sim\operatorname{Unif}([0,s)^2)$
independent of the inputs, define the periodic grid by
\[
  \mathcal C_V^{\mathbb T}
  :=\left\{(V+sj+[0,s)^2)\bmod1:
        j\in\{0,\ldots,s^{-1}-1\}^2\right\}.
\]
Since $s^{-1}$ is an integer, these $s^{-2}$ cells partition
$\mathbb T^2$, and each has area $s^2$.  A cell crossing a boundary of
$[0,1)^2$ continues from the opposite boundary and is counted as one cell,
even if its representation in $[0,1)^2$ consists of several pieces.

The following lemma transfers a lower bound on the total deficit
in the torus grid $\mathcal C_V^{\mathbb T}$ to the total deficit in
the square grid $\mathcal C_V$.

\begin{lemma}\label{lem:torus-square-deficit}
Fix any constant $c_{\rm def}>0$.
For every admissible rank $m$, if
\begin{equation}\label{eq:adaptive-deficit}
  \E\sum_{P\in\mathcal C_V^{\mathbb T}}(m|P|-N_P(B_m))_+
  \ge c_{\rm def}T_m,
\end{equation}
then
\begin{equation}\label{eq:square-uncut-deficit}
  \E\sum_{P\in\mathcal C_V}(m|P|-N_P(S_m))_+
  \ge \frac{c_{\rm def}}2T_m.
\end{equation}
Each expectation is over the inputs of the corresponding Greedy
process and the independent grid shift $V$.
\end{lemma}

By Lemma~\ref{lem:torus-square-deficit}, whose proof is deferred to
\Cref{sec:proof-torus-square-deficit}, it remains to prove the torus deficit bound \eqref{eq:adaptive-deficit}.

\paragraph{From request fluctuations to a torus server deficit.}
Fix an admissible rank $m$ and stop the torus Greedy process after matching
$X_1,\ldots,X_{n-m}$, when its remaining server set is $B_m$.
For each $1\le u\le n-m$, let $\mathsf G^{(-u)}$ denote
the torus Greedy process on the same initial servers that serves the
requests $X_1, \ldots, X_{u-1}, X_{u+1}, \ldots, X_{n-m}$.
After these requests, Lemma~\ref{lem:greedy-stability} gives its
unmatched server set as $B_m\cup\{Z_{m,u}\}$ for a unique extra server
$Z_{m,u}\notin B_m$.
Note that $Z_{m,u}$ need not be the server originally matched to $X_u$, since
skipping $X_u$ can change later matches.

Recall that we apply an independent grid shift $V\sim\operatorname{Unif}([0,s)^2)$, where $s=\lceil1/h_m\rceil^{-1}$.
Restoring $X_u$ adds a request in the cell of $X_u$
and removes a remaining server in the cell of $Z_{m,u}$.  If these cells
coincide, the request count and the server count change by opposite
amounts in the same cell.  Otherwise, we view the comparison as incurring an error.
The random-grid separation bound given in
\Cref{lem:random-grid-separation} controls the probability of this
error by $O(d_{\mathbb T}(X_u,Z_{m,u})/h_m)$.

Notice that past-request counts follow binomial distributions with expected total
absolute deviation $\Theta(\sqrt n/h_m)=\Theta(T_m)$.
The next lemma transfers these request-count fluctuations to
the remaining-server counts when the sum of the expected distances
from each skipped request $X_u$ to its extra server $Z_{m,u}$ is
sufficiently small.
Informally, this additional condition ensures that the total comparison-error probability remains small, and hence server-count fluctuations are comparable to request-count fluctuations.

\begin{samepage}
\begin{lemma}\label{prop:demand-localization}
There exists a universal constant $\varepsilon>0$ such that,
for every admissible rank $m$, if
\begin{equation}\label{eq:endpoint-first-hypothesis}
  \sum_{u=1}^{n-m}\E d_{\mathbb T}(X_u,Z_{m,u})
  \le\varepsilon n h_m,
\end{equation}
then
\begin{equation}\label{eq:localized-deficit}
  \E\sum_{P\in\mathcal C_V^{\mathbb T}}
       (m|P|-N_P(B_m))_+=\Omega(T_m).
\end{equation}
The expectation includes the inputs and the independent grid shift $V$.
\end{lemma}
\end{samepage}

We prove Lemma~\ref{prop:demand-localization} in \Cref{sec:proof-demand-localization}.
Its hypothesis \eqref{eq:endpoint-first-hypothesis} follows from
the next lemma, whose proof is deferred to \Cref{sec:proof-endpoint-localization}.

\begin{lemma}\label{prop:endpoint-localization}
For every $\sqrt n\le m<n$ and $1\le u\le n-m$,
\begin{equation}\label{eq:endpoint-first}
  \E d_{\mathbb T}(X_u,Z_{m,u})
  =O\!\left(\E D_{m+1}^{\mathbb T}\right)
  =O\!\left(\frac{\sqrt n}{m}\right).
\end{equation}
\end{lemma}

\paragraph{Putting everything together.}
To finish the proof of \Cref{thm:lb-2d}, sum
\eqref{eq:endpoint-first} over the choices of the skipped request $X_u$:
\[
  \sum_{u=1}^{n-m}\E d_{\mathbb T}(X_u,Z_{m,u})
  =O\!\left(\frac{n^{3/2}}{m}\right)
  =O\!\left(\frac{nh_m}{C_1}\right).
\]
The implicit constant is independent of $C_1$, so choosing $C_1$
sufficiently large ensures \eqref{eq:endpoint-first-hypothesis}.
Lemma~\ref{prop:demand-localization} then gives the torus
deficit \eqref{eq:adaptive-deficit}.
Lemma~\ref{lem:torus-square-deficit} gives the square deficit
bound \eqref{eq:square-uncut-deficit}, which satisfies
\eqref{eq:service-hypothesis} with
$d_0=\min\{c_{\rm def}/2,1\}$.
Lemma~\ref{prop:delayed-service} converts this square deficit
into the cost bound \eqref{eq:square-block-charge}.
The disjoint-interval argument at
the start of the section completes the proof.

\begin{remark}\label{rmk:critical-scale}
At rank $m$, a torus cell of side $s$ has expected remaining-server count
$ms^2$, whereas its past-request count has standard deviation
$\Theta(\sqrt n\,s)$.  These quantities balance at $s\asymp\sqrt n/m$, explaining the choice of $h_m$.
Lemma~\ref{prop:endpoint-localization} then permits the fluctuation
comparison at a sufficiently large constant multiple of this scale.
\end{remark}

\subsection{Proof of Lemma~\ref{prop:delayed-service}}
\label{sec:proof-delayed-service}

We reveal enough future requests to force many matches across grid
cells in expectation, then average over the independent grid shift to
lower-bound their total cost.

Define $r:=\lfloor d_0T_m/2\rfloor < m$ and reveal the next $m-r$ requests,
$X_{n-m+1},\ldots,X_{n-r}$.
The Greedy process
starts with the square server configuration $S_m$ and has $r$ unmatched
servers after these requests have been matched.
For each cell $P\in\mathcal C_V$, let $F_P$ denote the number
of these requests whose locations lie in $P$:
\[
  F_P:=\sum_{i=n-m+1}^{n-r}\mathbf1_{\{X_i\in P\}}.
\]
Intersecting every square of the shifted grid with $[0,1]^2$ gives
a partition consisting of the full cells in $\mathcal C_V$ and the
nonempty boundary portions.
Let $N_{\rm cross}$ denote the number of matches made while serving
these requests for which the request and its assigned server lie in
different cells of this partition.

Among the $F_P$ requests in a cell $P\in\mathcal C_V$, at most
$N_P(S_m)$ of them can be matched to a server in $P$.  Each remaining
request contributes a distinct match to $N_{\rm cross}$.  Thus, in every sample path,
\begin{equation}\label{eq:flux-identity}
  N_{\rm cross}\ge
  \sum_{P\in\mathcal C_V}(F_P-N_P(S_m))_+.
\end{equation}
Condition on $\mathcal H_m:=\sigma(V,Y_1,\ldots,Y_n,X_1,\ldots,X_{n-m})$,
which fixes the grid and $S_m$ while leaving the future requests
independent and uniform.
Thus, for every $P\in\mathcal C_V$,
\begin{equation}\label{eq:origin-noise}
  \E[F_P\mid\mathcal H_m]=(m-r)|P|.
\end{equation}
Jensen's inequality and \eqref{eq:origin-noise}, followed by
$(a-b)_+\ge a_+-b$ for $b\ge0$, give
\begin{equation}\label{eq:residual-deficit}
\begin{aligned}
  \E\bigl[(F_P-N_P(S_m))_+\mid\mathcal H_m\bigr]
  \ge ((m-r)|P|-N_P(S_m))_+
  \ge (m|P|-N_P(S_m))_+-r|P|.
\end{aligned}
\end{equation}
Taking expectations in \eqref{eq:flux-identity}, applying
\eqref{eq:residual-deficit}, and using
$\sum_{P\in\mathcal C_V}|P|\le1$ and \eqref{eq:service-hypothesis}
now yield
\begin{equation}\label{eq:deficit-flux}
\begin{aligned}
  \E N_{\rm cross}
  \ge\E\sum_{P\in\mathcal C_V}(m|P|-N_P(S_m))_+-r
  \ge d_0T_m-r\ge\frac{d_0T_m}{2}.
\end{aligned}
\end{equation}

To upper-bound $\E N_{\rm cross}$, condition on all
server and request locations.  All matches by Greedy are then fixed,
while the independent grid shift $V$ remains uniform.
Recall that the matches of requests
$X_{n-m+1},\ldots,X_{n-r}$ have total cost $G_{[r+1,m]}$.
Applying Lemma~\ref{lem:random-grid-separation} to each match,
summing over these matches, and averaging over the input locations gives
\begin{equation}\label{eq:flux-cost}
\E N_{\rm cross}=O\!\left(\frac1{h_m}\E G_{[r+1,m]}\right).
\end{equation}
Combining \eqref{eq:deficit-flux} and \eqref{eq:flux-cost} gives
\begin{align*}
\E G_{[r+1,m]}=\Omega(d_0h_mT_m)=\Omega(\sqrt n),
\end{align*}
where $h_mT_m=\sqrt n$ by \eqref{eq:adaptive-scale}.
Choose the constant $c$ in the lemma at most $d_0/2$.
Then $r+1>d_0T_m/2\ge cT_m$, so $[r+1,m]\subseteq[cT_m,m]$ and
\eqref{eq:service-conclusion} follows.

\subsection{Proof of Lemma~\ref{lem:torus-square-deficit}}
\label{sec:proof-torus-square-deficit}

Use the coupling from \Cref{sec:random-cut-discrepancy},
with the random cut $U$ independent of the torus inputs and the grid
shift $V$.
Let $\mathcal C_{\rm uncut}(U,V)$ consist of the cells
$P\in\mathcal C_V^{\mathbb T}$ whose interiors avoid the cut lines
$z_1=U_1$ and $z_2=U_2$.
Note that each cut line meets the interiors of at most $s^{-1}$ grid cells,
each of area $s^2$.
The excluded cells therefore have total area at most
$O(s)\le O(h_m)$.
Since $(m|P|-N_P(B_m))_+\le m|P|$ for every cell $P$,
the total deficit of the excluded cells is at most $O(m h_m)$.
For each cell $P\in\mathcal C_V^{\mathbb T}$, compare the numbers
$N_P(A_m)$ and $N_P(B_m)$ of remaining servers in that cell.
Servers in $A_m\cap B_m$ contribute equally to these two numbers.
Since $A_m\setminus B_m$ and $B_m\setminus A_m$ each contain $K_m$
servers, summing the absolute differences over all cells gives
\[
  \sum_{P\in\mathcal C_V^{\mathbb T}}
    \bigl|\,N_P(A_m)-N_P(B_m)\,\bigr|\le2K_m.
\]
For each retained cell $P\in\mathcal C_{\rm uncut}(U,V)$, the inequality
$|x_+-y_+|\le|x-y|$ gives
\[
  (m|P|-N_P(A_m))_+
  \ge (m|P|-N_P(B_m))_+
  -\bigl|N_P(A_m)-N_P(B_m)\bigr|.
\]
Summing over the retained cells and using the preceding bound on count
differences yields
\[
\begin{aligned}
  \sum_{P\in\mathcal C_{\rm uncut}(U,V)}(m|P|-N_P(A_m))_+
  \ge
  \sum_{P\in\mathcal C_{\rm uncut}(U,V)}(m|P|-N_P(B_m))_+-2K_m.
\end{aligned}
\]
Combining this with the deficit bound on the excluded cells gives,
in every sample path,
\begin{equation}\label{eq:deficit-transfer}
  \sum_{P\in\mathcal C_{\rm uncut}(U,V)}(m|P|-N_P(A_m))_+
  \ge \sum_{P\in\mathcal C_V^{\mathbb T}}(m|P|-N_P(B_m))_+-O(mh_m+K_m).
\end{equation}
Recall from \eqref{eq:rank-cutoff} that $M_n=\sqrt n\log n$
when $d=2$.
Lemma~\ref{lem:discrepancy} therefore gives
\[
  \E K_m\le O(M_n)=O(\sqrt n\log n).
\]
Substituting $h_m=\Theta(\sqrt n/m)$ and $T_m=\Theta(m)$ from
\eqref{eq:intro-localization-scale} gives
\begin{align*}
  \frac{mh_m+\E K_m}{T_m}
  \le O \left( \frac{\sqrt n+ \sqrt n\log n}{m} \right)
  \le O(n^{-1/4}+n^{-1/4}\log n),
\end{align*}
where the last inequality uses $m\ge n^{3/4}$.
Thus, for all sufficiently large $n$, the expected deficit lost
in \eqref{eq:deficit-transfer} is at most $c_{\rm def}T_m/2$ for every
admissible rank $m$.
Taking expectations in \eqref{eq:deficit-transfer}
and using \eqref{eq:adaptive-deficit} therefore gives the intermediate bound
\[
  \E\sum_{P\in\mathcal C_{\rm uncut}(U,V)}(m|P|-N_P(A_m))_+
  \ge\frac{c_{\rm def}}2T_m.
\]

We now transfer this intermediate bound to the square grid and the original
square Greedy process.
By Lemma~\ref{lem:unwrapped-greedy-law}, conditional on $U$,
unwrapping the Greedy process under $d_U$ gives the joint law of the
original square inputs and Greedy process.
Unwrapping translates the grid by $-U$, giving the square grid
shift $\widehat V=(V-U)\pmod s$.
Conditional on $U$, this shift is uniform on $[0,s)^2$ and
independent of the unwrapped inputs.
This conditional joint law does not depend on $U$.  Thus the
unwrapped process with grid shift $\widehat V$ has the same joint
law as $S_k$ with an independent uniform grid shift $V$.
Because $s^{-1}$ is an integer, every uncut torus cell maps to a full
cell of this square grid, preserving its area and remaining-server count.
Every deficit term is nonnegative, so including all full square
grid cells and using the intermediate bound above gives
\begin{align*}
  \E\sum_{Q\in\mathcal C_V}(m|Q|-N_Q(S_m))_+
  &=\E\sum_{Q\in\mathcal C_{\widehat V}}
    (m|Q|-N_Q(\kappa_U(A_m)))_+\\
  &\ge \E\sum_{P\in\mathcal C_{\rm uncut}(U,V)}
    (m|P|-N_P(A_m))_+
  \ge\frac{c_{\rm def}}2T_m,
\end{align*}
concluding \eqref{eq:square-uncut-deficit}.

\subsection{Proof of Lemma~\ref{prop:demand-localization}}
\label{sec:proof-demand-localization}

Let
\[
  L:=(N_P(B_m))_{P\in\mathcal C_V^{\mathbb T}},
  \qquad \bar L:=(m|P|)_{P\in\mathcal C_V^{\mathbb T}}.
\]
Thus $L$ records the remaining-server counts, and $\bar L$ is their
uniform benchmark.
Both vectors have total mass $m$, so the positive and negative
parts of $L-\bar L$ have equal total mass, giving the deficit identity
\begin{equation}\label{eq:torus-deficit-identity}
  \sum_{P\in\mathcal C_V^{\mathbb T}}(m|P|-N_P(B_m))_+
  =\frac12\|L-\bar L\|_1.
\end{equation}

For a point $x$, let $e_V(x)$ be its cell indicator vector, whose
coordinate at $P$ is $\mathbf1_{\{x\in P\}}$.
Let $P_V(x)$ be the cell containing $x$.
Define the centered request-count vector by
\[
  W:=\sum_{u=1}^{n-m}\left(e_V(X_u)-\frac{\bar L}{m}\right).
\]
Add the remaining-server and past-request counts to form
\[
  H:=L+\sum_{u=1}^{n-m}e_V(X_u).
\]
We will show that the expected $\ell_1$ deviation of $H$ from
its mean given the initial servers and grid is much smaller than
$\E\|W\|_1$.
The remaining-server counts must therefore fluctuate enough
to offset the request-count fluctuations.

\paragraph{Step 1: the effect of skipping one request.}
Fix $1\le u\le n-m$ and consider $\mathsf G^{(-u)}$ through
slot $n-m$.
Its remaining-server-count vector is $L+e_V(Z_{m,u})$, by the
definition of $Z_{m,u}$.
Its combined count vector is therefore
\[
  H^{(-u)}:=L+e_V(Z_{m,u})+
       \sum_{\substack{1\le v\le n-m\\v\ne u}}e_V(X_v).
\]
Since $\mathsf G^{(-u)}$ never uses $X_u$, the vector
$H^{(-u)}$ is a function only of the initial servers, the grid, and
the other past requests.
It holds that
\begin{equation}\label{eq:skipped-request-conditional-mean}
  H-H^{(-u)}=e_V(X_u)-e_V(Z_{m,u}).
\end{equation}
Hence, whenever $X_u$ and $Z_{m,u}$ lie in the same grid cell, we have $H = H^{(-u)}$.

\paragraph{Step 2: the combined counts fluctuate little.}
We condition on $\mathcal F_0:=\sigma(Y_1,\ldots,Y_n,V)$,
which fixes the initial server locations and the grid shift.
Since $X_1,\ldots,X_{n-m}$ are independent of
$(Y_1,\ldots,Y_n,V)$, their conditional law given $\mathcal F_0$
is still that of independent uniform torus points.
Under this conditional law, apply Lemma~\ref{lem:efron-stein}(ii)
to each scalar cell count $H_P$ with comparison functions $H_P^{(-u)}$,
using the independent inputs $X_1,\ldots,X_{n-m}$.
These counts are bounded, and Step~1 shows that each $H_P^{(-u)}$
does not depend on $X_u$.
Summing the bounds~\eqref{eq:efron-stein-deletion} over the grid
cells and then averaging over the initial servers and grid gives
\begin{equation}\label{eqn:ortho-martingale-diff}
  \E\bigl\|H-\E[H\mid\mathcal F_0]\bigr\|_2^2
  \le\sum_{u=1}^{n-m}
       \E\|H-H^{(-u)}\|_2^2.
\end{equation}%
By \eqref{eq:skipped-request-conditional-mean}, each squared
norm on the right is zero when the two points lie in the same cell
and equals two otherwise, so \eqref{eqn:ortho-martingale-diff} yields
\begin{equation}\label{eq:remainder-two}
  \E\bigl\|H-\E[H\mid\mathcal F_0]\bigr\|_2^2
  \le2\sum_{u=1}^{n-m}
       \Pp\{P_V(Z_{m,u})\ne P_V(X_u)\}.
\end{equation}%
Each pair $(X_u,Z_{m,u})$ is independent of the grid shift $V$,
so Lemma~\ref{lem:random-grid-separation} and
\eqref{eq:endpoint-first-hypothesis} turn \eqref{eq:remainder-two} into
\[
  \E\bigl\|H-\E[H\mid\mathcal F_0]\bigr\|_2^2
  =O\!\left(\frac1{h_m}\sum_{u=1}^{n-m}
       \E d_{\mathbb T}(X_u,Z_{m,u})\right)
  =O(\varepsilon n).
\]%
Applying the Cauchy--Schwarz inequality first over the $s^{-2}$ grid cells
and then to the expectation over all inputs and the grid shift $V$
yields
\begin{equation}\label{eq:comparison-error-one}
\begin{aligned}
  \E\bigl\|H-\E[H\mid\mathcal F_0]\bigr\|_1
  &\le s^{-1}\E\bigl\|H-\E[H\mid\mathcal F_0]\bigr\|_2\\
  &\le s^{-1}\Bigl(\E\bigl\|H-\E[H\mid\mathcal F_0]\bigr\|_2^2\Bigr)^{1/2}\\
  &=O(s^{-1}\sqrt{\varepsilon n})
   =O(\sqrt\varepsilon\,T_m),
\end{aligned}
\end{equation}
where the last line uses the preceding $O(\varepsilon n)$
second-moment bound, $s=\Theta(h_m)$, and $T_m=\sqrt n/h_m$.

\paragraph{Step 3: the request fluctuations force a server deficit.}
Notice that, conditional on $V$, each $W_P$ is a centered binomial variable with
parameters $n-m$ and $s^2$.
For every admissible $m$, we have $n-m=\Theta(n)$ and
$s=\Theta(h_m)=o(1)$, so
\[
  \sigma^2:=\E W_P^2=(n-m)s^2(1-s^2)=\Theta(nh_m^2).
\]%
Since $nh_m^2=C_1^2n^2/m^2\ge C_1^2/c_0^2$, choosing the
fixed constant $C_1$ sufficiently large ensures $\sigma^2\ge1$.
The binomial fourth-moment formula then gives
\[
  \E W_P^4=O(\sigma^4+\sigma^2)=O(\sigma^4).
\]%
H\"older's inequality now yields
\[
  \E|W_P|\ge
  \frac{(\E W_P^2)^{3/2}}{(\E W_P^4)^{1/2}}
  =\Omega(\sigma)=\Omega(\sqrt n\,h_m).
\]%
Summing over the $s^{-2}$ cells gives
\begin{equation}\label{eq:request-fluctuation-one}
  \E\|W\|_1=\Omega(s^{-2}\sqrt n\,h_m)=\Omega(T_m).
\end{equation}%
To compare these request fluctuations with the remaining-server counts,
note that each request satisfies
$\E[e_V(X_u)\mid\mathcal F_0]=\bar L/m$, so
$\E[W\mid\mathcal F_0]=0$.
Since $H=L+W+(n-m)\bar L/m$ and $\bar L$ is
$\mathcal F_0$-measurable, subtracting conditional means and
rearranging gives
\begin{equation}\label{eq:doob-sum}
  W=\bigl(H-\E[H\mid\mathcal F_0]\bigr)
      -(L-\bar L)+\E[L-\bar L\mid\mathcal F_0].
\end{equation}%
Taking norms in \eqref{eq:doob-sum}, applying the triangle
inequality, and then taking expectations yields
\begin{align*}
  \E\|W\|_1
  &\le\E\bigl\|H-\E[H\mid\mathcal F_0]\bigr\|_1
       +\E\|L-\bar L\|_1
  +\E\bigl\|\E[L-\bar L\mid\mathcal F_0]\bigr\|_1\\
  &\le\E\bigl\|H-\E[H\mid\mathcal F_0]\bigr\|_1
       +2\E\|L-\bar L\|_1,
\end{align*}
where the last inequality is conditional Jensen for the convex
function $z\mapsto\|z\|_1$, followed by taking expectations.
Combining \eqref{eq:comparison-error-one} with the request
fluctuation bound \eqref{eq:request-fluctuation-one}, for
sufficiently small universal $\varepsilon$, gives
\[
  2\E\|L-\bar L\|_1
  \ge\E\|W\|_1
     -\E\bigl\|H-\E[H\mid\mathcal F_0]\bigr\|_1
  =\Omega(T_m).
\]%
The deficit identity \eqref{eq:torus-deficit-identity} now gives \eqref{eq:localized-deficit}.

\subsection{Proof of \texorpdfstring{\Cref{prop:endpoint-localization}}{Lemma~\ref{prop:endpoint-localization}}}
\label{sec:proof-endpoint-localization}

We first analyze a fixed order of server deletions, then apply
the resulting bound to the deletion order in $\mathsf G^{(-u)}$.
Fix a finite set $\Sigma_0\subset\mathbb T^2$ of distinct
server locations and an ordered list $r_0,\ldots,r_{N-1}$ of
distinct servers in $\Sigma_0$, with $0\le N<|\Sigma_0|$.
This list is a \emph{deletion sequence}: step $i$ removes
server $r_i$.
Let $\Sigma_i:=\Sigma_0\setminus\{r_j:0\le j<i\}$ be the
set of servers remaining after $i$ deletions and $F:=\Sigma_N$,
so
\[
  \Sigma_0\supset\Sigma_1\supset\cdots\supset\Sigma_N=F\ne\varnothing,
  \qquad \Sigma_{i+1}=\Sigma_i\setminus\{r_i\},\quad 0\le i<N.
\]%
For $0\le i\le N$ and $s\in\Sigma_i$, let $C_i(s)$ be the torus Voronoi cell of $s$ in $\Sigma_i$ and
$v_i(s):=|C_i(s)|$.
Sample $X$ uniformly on the torus and $W_i$ uniformly on $C_i(r_i)$,
$0\le i<N$, with all these variables mutually independent.
The choice $W_i\in C_i(r_i)$ ensures that Greedy selects the
prescribed server $r_i$ at step $i$.
To model the extra server, we place a mark on one server and
update its position as servers are removed.
Let $L_i\in\Sigma_i$ denote the server carrying the mark at step $i$,
and define its evolution by
\[
  L_0=\NN_{\Sigma_0}^{d_{\mathbb T}}(X),\qquad
  L_{i+1}=
  \begin{cases}
    L_i,&L_i\ne r_i,\\
    \NN_{\Sigma_{i+1}}^{d_{\mathbb T}}(W_i),&L_i=r_i.
  \end{cases}
\]
The mark records the location $L_i$ of the extra server and stays at its current server until that server is deleted.
When the deleted server carries the mark, the mark moves to the
nearest surviving server to the same request $W_i$.

Although the mark may change servers repeatedly, the following
lemma bounds its expected final distance from $X$ by a universal
constant times the average distance to the nearest server in $F$.

\begin{lemma}\label{lem:fixed-deletion-localization}
For every fixed deletion sequence and marked process defined above,
\begin{equation}\label{eq:fixed-deletion-first}
  \E d_{\mathbb T}(X,L_N)
  =O\!\left(\int_{\mathbb T^2}d_{\mathbb T}(x,F)\,\mathrm{d}x\right),
\end{equation}
where the expectation is over $X,W_0,\ldots,W_{N-1}$.
\end{lemma}

The proof of \Cref{lem:fixed-deletion-localization} is given in Section~\ref{sec:proof-fixed-deletion-localization}, and we now finish the proof of Lemma~\ref{prop:endpoint-localization}.
Fix the index $u$ of the request to be skipped.
Condition on the initial server configuration and the requests
$X_1,\ldots,X_{u-1}$.
Let $\Sigma_0:=B_{n-u+1}$ be the set of servers still available
in both the original torus Greedy process and $\mathsf G^{(-u)}$
after they have served $X_1,\ldots,X_{u-1}$.
Set $N=n-m-u$ and condition on the identity of the server
matched by $\mathsf G^{(-u)}$ to each of the $N$ requests
$X_{u+1},\ldots,X_{n-m}$.
For $0\le i<N$, write $r_i$ for the server matched to
$X_{u+i+1}$.
The resulting fixed list determines each remaining set
$\Sigma_i=\Sigma_0\setminus\{r_j:0\le j<i\}$ and hence the
Voronoi cell $C_i(r_i)$.
Given the initial servers and the requests before slot $u$,
specifying this list of matches is exactly the event
\[
  \bigcap_{i=0}^{N-1}\{X_{u+i+1}\in C_i(r_i)\}.
\]%
Each event restricts a different independent request to a fixed Voronoi cell.
Thus the conditional laws of $W_i:=X_{u+i+1}$ are independent and
uniform on $C_i(r_i)$.
The skipped request $X_u$ remains independent and uniform,
because $\mathsf G^{(-u)}$ never uses it.

At slot $u$, the original torus Greedy process matches $X_u$
to its nearest server in $\Sigma_0$ and removes that server, while
$\mathsf G^{(-u)}$ skips $X_u$ and leaves $\Sigma_0$ unchanged.
Thereafter $\mathsf G^{(-u)}$ has one extra unmatched server
compared with the original torus process.
This extra server stays unchanged unless $\mathsf G^{(-u)}$
matches it to a later request; in that case, the new extra server is
the one chosen by the original torus process for that same request.
Thus, under the stated conditioning, its location evolves exactly
as the marked server in Lemma~\ref{lem:fixed-deletion-localization},
and its final location is $Z_{m,u}$.

Let $F_{m,u}:=B_m\cup\{Z_{m,u}\}$ be the unmatched server set
of $\mathsf G^{(-u)}$ after slot $n-m$, which has $m+1$ servers.
Note that $F_{m,u}$ has the same law as $B_{m+1}$ since
$\mathsf G^{(-u)}$ has served $n-m-1$ iid uniform requests.
An independent uniform request therefore has expected distance
$\E D_{m+1}^{\mathbb T}$ to $F_{m,u}$.
Applying \Cref{lem:fixed-deletion-localization} under the preceding
conditioning and then averaging gives
\begin{align*}
  \E d_{\mathbb T}(X_u,Z_{m,u})
  =O\!\left(\E\int_{\mathbb T^2}
       d_{\mathbb T}(x,F_{m,u})\,\mathrm{d}x\right)
  =O\!\left(\E D_{m+1}^{\mathbb T}\right)
   =O\!\left(\frac{\sqrt n}{m}\right),
\end{align*}%
concluding \eqref{eq:endpoint-first}.

\subsection{Proof of \Cref{lem:fixed-deletion-localization}}
\label{sec:proof-fixed-deletion-localization}

We first show that the distribution of the marked server at each step is induced by the Voronoi-cell volumes of the existing servers.

\begin{claim}\label{clm:mark-volume-law}
For all $0 \leq i \leq N$ and $s \in \Sigma_i$,
\begin{equation}\label{eq:killed-mark-domination}
  \Pp\{L_i=s\}=v_i(s),
\end{equation}
where the probability is taken over $X,W_0,\ldots,W_{i-1}$
for the fixed deletion sequence.
\end{claim}

\begin{proof}
For $i = 0$, \eqref{eq:killed-mark-domination} holds since $X$ is uniform.
For each server $s \in \Sigma_{i+1}$, deleting $r_i$ expands the Voronoi cell $C_i(s)$ by
$C_i(r_i)\cap C_{i+1}(s)$.
Assuming \eqref{eq:killed-mark-domination} holds at step $i$, the independence of $W_i$ and $L_i$ gives,
for $s\in\Sigma_{i+1}$,
\[
\begin{aligned}
  \Pp\{L_{i+1}=s\}
  =v_i(s)+v_i(r_i)
      \frac{|C_i(r_i)\cap C_{i+1}(s)|}{v_i(r_i)}
  =v_i(s)+|C_i(r_i)\cap C_{i+1}(s)|
   =v_{i+1}(s),
\end{aligned}
\]
as desired.
\end{proof}

For each $x \in \mathbb T^2$, let $\rho(x):=d_{\mathbb T}(x,F)$, and let $B_{\mathbb T}(x,r)$ be the closed torus ball of radius $r$ centered at $x$.
Also, define
\begin{equation}\label{eq:fixed-deletion-enlarged-radius}
  a(x):=\sup\{\rho(y):y\in\mathbb T^2,
       d_{\mathbb T}(x,y)\le2\rho(y)\}.
\end{equation}
To see why the above definition is useful, we claim that, for all $x \in \mathbb T^2$ and server $s \in \Sigma_i$, we have
\begin{align}
    \max_{y \in C_i(s)} \rho(y)
    \leq \max\{d_{\mathbb T}(x,s),a(x)\}.
    \label{eqn:property-a-x}
\end{align}
To see this, if $\rho(y) > d_{\mathbb T}(x,s)$ for $y \in C_i(s)$, since $F \subseteq \Sigma_i$,
\begin{align*}
    d_{\mathbb T}(x,y)
  \le d_{\mathbb T}(x,s)+d_{\mathbb T}(s,y)
  \le d_{\mathbb T}(x,s)+\rho(y)<2\rho(y),
\end{align*}
and hence $y$ is admissible in the definition of $a(x)$.
Importantly, $a(x)$ does not depend on $i$ and $s$.

By taking $y=x$ in the definition of $a(x)$, we know that $a(x)\ge\rho(x)$.
Although $a(x)$ is an overestimate of $\rho(x)$, the following lemma shows that this enlargement only increases
the spatial average by a constant factor.

\begin{lemma}\label{lem:fixed-deletion-covering}
It holds that
\[
\int_{\mathbb T^2}a(x)\,\mathrm{d}x
  \le O\!\left(\int_{\mathbb T^2}\rho(x)\,\mathrm{d}x\right).
\]
\end{lemma}

\begin{proof}
Fix $x$ and $y$ with $d_{\mathbb T}(x,y)\le2\rho(y)$
and $\rho(y)>0$.
For every $z\in B_{\mathbb T}(y,\rho(y)/2)$, the $1$-Lipschitz property of $\rho$ and the triangle inequality give
\[
\frac{\rho(y)}2\le\rho(z)\le\frac{3\rho(y)}2,\qquad
  d_{\mathbb T}(x,z)\le\frac{5\rho(y)}2\le5\rho(z),
\]
implying that
\begin{align*}
    B_{\mathbb T}(y,\rho(y)/2)
    \subseteq \{z:d_{\mathbb T}(x,z)\le5\rho(z)\}.
\end{align*}
Consequently,
\[
\begin{aligned}
  \int_{\{z:d_{\mathbb T}(x,z)\le5\rho(z)\}}
  \frac{\mathrm{d}z}{\rho(z)}
  \ge\Omega \left(\frac{\,|B_{\mathbb T}(y,\rho(y)/2)|}{\rho(y)} \right)
  =\Omega(\rho(y)),
\end{aligned}
\]
where the equality holds since $|B_{\mathbb T}(y,\rho(y)/2)|=\Theta(\rho(y)^2)$.
For fixed $x$, the integral on the left is independent of the
admissible choice of $y$.
Taking the supremum over these choices in
\eqref{eq:fixed-deletion-enlarged-radius} therefore gives
\[
a(x)\le O\!\left(\int_{\mathbb T^2}
  \frac{\mathbf1_{\{d_{\mathbb T}(x,z)\le5\rho(z)\}}}{\rho(z)}
  \,\mathrm{d}z\right).
\]
Integrating this pointwise bound over $x$, it follows that
\[
\begin{aligned}
    \int_{\mathbb T^2}a(x)\,\mathrm{d}x
  &\leq O \left( \int_{\mathbb T^2}\int_{\mathbb T^2}
    \frac{\mathbf1_{\{d_{\mathbb T}(x,z)\le5\rho(z)\}}}{\rho(z)}
    \,\mathrm{d}z\,\mathrm{d}x \right)
  = O \left( \int_{\mathbb T^2}\frac1{\rho(z)}
    \left(\int_{\mathbb T^2}
      \mathbf1_{\{d_{\mathbb T}(x,z)\le5\rho(z)\}}
      \,\mathrm{d}x\right)\mathrm{d}z \right)\\
  &= O \left( \int_{\mathbb T^2}
    \frac{|B_{\mathbb T}(z,5\rho(z))|}{\rho(z)}\,\mathrm{d}z \right)
  \le O\!\left(\int_{\mathbb T^2}\rho(x)\,\mathrm{d}x\right),
\end{aligned}
\]
where the first equality holds by applying Tonelli's theorem, and the last inequality holds since $|B_{\mathbb T}(z,t)|=O(t^2)$.
This concludes the proof.
\end{proof}

Now, we are ready to prove Lemma~\ref{lem:fixed-deletion-localization}.
For a point $x$ and a server $s$, define the potential
\[
\Phi_x(s):=\sqrt{d_{\mathbb T}(x,s)^2+a(x)^2},\qquad
  \Phi_i:=\Phi_X(L_i).
\]
By \eqref{eqn:property-a-x}, the potential dominates $\rho$
throughout the Voronoi cell of $s$:
\begin{equation}\label{eq:fixed-deletion-potential-cell}
\rho(y)\le\max\{d_{\mathbb T}(x,s),a(x)\}
  \le\Phi_x(s),\qquad y\in C_i(s).
\end{equation}

For $0\le i<N$, define $\Delta_i(y):=d_{\mathbb T}(y,\Sigma_{i+1})^2-d_{\mathbb T}(y,\Sigma_i)^2\ge0$.
Note that $\Delta_i(\cdot)$ vanishes outside the deleted cell $C_i(r_i)$, and the request $W_i$ is independent of $(X,L_i)$.
On the event $\{L_i=r_i\}$, let $u$ be the vector from $X$ to a nearest
periodic copy of $r_i$, so that $\|u\|_2=d_{\mathbb T}(X,r_i)$.
We use the notations $\mathcal V_{r_i}$ and $T_{r_i}$ from
\Cref{sec:periodic-voronoi} for state $\Sigma_i$.
We can write $W_i=r_i+z\pmod1$, where $z$ is conditionally uniform in
$\mathcal V_{r_i}$ given $(X,L_i)$.
Then, $L_{i+1}=r_i+T_{r_i}(z)\pmod1$.
Figure~\ref{fig:marked-update-periodic-lift} illustrates these
points and vectors using chosen periodic copies.

\begin{figure}[H]
  \centering
  % Exact unit-torus example, rotated by 90 degrees for a compact view.
% Initial server representatives relative to r_i are (0,0) and (-0.2,0.4).
% Chosen copies: r_i=L_i=(0,0), X=(-0.15,-0.4),
% W_i=z=(0.1,0.08), and L_{i+1}=T_{r_i}(z)=(-0.2,0.4).
% The polygon below is the exact lifted Voronoi cell of the origin.
% X lies outside it, while z lies inside it. The surviving copy shown
% is nearest to z; the copy ( -0.2,-0.6 ) is nearer to X instead.
\begingroup
\colorlet{markedink}{.}
\begin{tikzpicture}[
  x=5cm,y=5cm,
  font=\footnotesize,
  draw=markedink,text=markedink,
  line cap=round,line join=round,>=stealth
]
  \coordinate (markedOrigin) at (0,0);
  \coordinate (markedStart) at (-0.15,-0.4);
  \coordinate (markedRequest) at (0.1,0.08);
  \coordinate (markedSurvivor) at (-0.2,0.4);

  \filldraw[fill=markedink!7,line width=0.55pt]
    (0.5,-0.166667) -- (0.5,0) -- (0.3,0.4)
    -- (-0.5,0) -- (-0.5,-0.166667)
    -- (0.3,-0.433333) -- cycle;
  \node[inner sep=1pt] at (0.32,-0.13) {$\mathcal V_{r_i}$};

  % The surviving copy is chosen relative to the request, not to X.
  \draw[line width=0.25pt,draw=markedink!35]
    (markedRequest) -- (markedSurvivor);

  \draw[->,line width=0.85pt] (markedStart) -- (markedOrigin)
    node[pos=0.48,right,inner sep=2pt] {$u$};
  \draw[->,line width=0.85pt] (markedOrigin) -- (markedRequest)
    node[pos=0.54,above left,inner sep=1.5pt] {$z$};
  \draw[->,line width=0.85pt] (markedOrigin) -- (markedSurvivor);
  \draw[->,dashed,line width=0.85pt] (markedStart) -- (markedSurvivor);

  % Transparent horizontal labels outside the cell point to their vectors.
  \draw[line width=0.35pt,draw=markedink!65]
    (-0.12,0.24) -- (0,0.43);
  \node[anchor=west,inner sep=1pt] at (0.02,0.43) {$T_{r_i}(z)$};
  \draw[line width=0.35pt,draw=markedink!65]
    (-0.175,0) -- (-0.65,-0.04);
  \node[anchor=east,inner sep=1pt] at (-0.68,-0.04) {$u+T_{r_i}(z)$};

  \fill (markedOrigin) circle[radius=1.9pt];
  \node[below right,inner sep=3pt] at (markedOrigin) {$r_i=L_i$};
  \node[rectangle,fill=markedink,inner sep=0pt,minimum size=3.8pt]
    at (markedStart) {};
  \node[below left,inner sep=3pt] at (markedStart) {$X$};
  \filldraw[fill=white,line width=0.65pt]
    (markedRequest) circle[radius=1.8pt];
  \node[right,inner sep=3pt] at (markedRequest) {$W_i$};
  \fill (markedSurvivor) circle[radius=1.9pt];
  \node[above,inner sep=3pt] at (markedSurvivor) {$L_{i+1}$};
\end{tikzpicture}
\endgroup
  \caption{One marked-server update on the event $L_i=r_i$, shown using
  chosen periodic copies in $\mathbb R^2$.
  The chosen copy of $r_i$ is at the origin; the copies of $X$, $W_i$,
  and $L_{i+1}$ are at $-u$, $z\in\mathcal V_{r_i}$, and
  $T_{r_i}(z)$, respectively.
  The copy of $X$ may lie outside the shaded cell $\mathcal V_{r_i}$.
  The dashed vector is $u+T_{r_i}(z)$. The selected copy for $L_{i + 1}$ is
  nearest to $z$ but need not be nearest to $-u$, so
  $d_{\mathbb T}(X,L_{i+1})\le\|u+T_{r_i}(z)\|_2$.}
  \label{fig:marked-update-periodic-lift}
\end{figure}

The potential difference is bounded by
\[
\begin{aligned}
  \Phi_{i+1}-\Phi_i
  \le\sqrt{\|u+T_{r_i}(z)\|_2^2+a(X)^2}-\Phi_i
  \le\frac{2u\cdot T_{r_i}(z)+\|T_{r_i}(z)\|_2^2}{2\Phi_i},
\end{aligned}
\]
where the second inequality holds since $\sqrt B-\sqrt A \leq (B - A) / 2 \sqrt{A}$ for $A > 0$ and $B \geq 0$.
Conditional on $X=x$ and $L_i=r_i$, the vector $u$ and the
denominator $\Phi_x(r_i)$ are fixed.
As a result,
\[
\begin{aligned}
  \E[\Phi_{i+1}-\Phi_i\mid X=x,L_i=r_i]
  &\le\frac{1}{2v_i(r_i)\Phi_x(r_i)}
    \int_{\mathcal V_{r_i}}\|T_{r_i}(z)\|_2^2\,\mathrm{d}z\\
  &\le\frac{3}{2v_i(r_i)}\int_{C_i(r_i)}
    \frac{\Delta_i(y)}{\Phi_x(r_i)}\,\mathrm{d}y
  \le\frac{3}{2v_i(r_i)}\int_{C_i(r_i)}
    \frac{\Delta_i(y)}{\rho(y)}\,\mathrm{d}y,
\end{aligned}
\]
where the first inequality holds since the first identity from Lemma~\ref{lem:periodic-voronoi-deletion} cancels the linear term, the second inequality holds by the second-moment bound from Lemma~\ref{lem:periodic-voronoi-deletion}, and the last inequality holds by \eqref{eq:fixed-deletion-potential-cell}.
On the other hand, when $L_i\ne r_i$, the marked server and the potential remain unchanged.
Averaging the conditional bound over $X$ given $L_i=r_i$ and
multiplying by $\Pp\{L_i=r_i\}$ gives
\begin{equation}\label{eq:fixed-deletion-potential-drift}
\begin{aligned}
  \E[\Phi_{i+1}-\Phi_i]
  \le O\!\left( \frac{\Pp\{L_i=r_i\}}{v_i(r_i)}
       \int_{C_i(r_i)}
         \frac{\Delta_i(y)}{\rho(y)}\,\mathrm{d}y\right)
  \leq O\!\left(\int_{\mathbb T^2}
         \frac{\Delta_i(y)}{\rho(y)}\,\mathrm{d}y\right),
\end{aligned}
\end{equation}
where the second inequality holds by Claim~\ref{clm:mark-volume-law}.
Also, $d_{\mathbb T}(X,L_0)\le\rho(X)\le a(X)$ implies
$\Phi_0\le\sqrt2\,a(X)$.
Summing \eqref{eq:fixed-deletion-potential-drift} gives
\[
\begin{aligned}
  \E d_{\mathbb T}(X,L_N)
  \le\E\Phi_N
  \le\sqrt2\int_{\mathbb T^2}a(x)\,\mathrm{d}x+
    O\!\left(\int_{\mathbb T^2}
      \frac{\sum_{i=0}^{N-1}\Delta_i(y)}{\rho(y)}\,\mathrm{d}y\right)
  \le O\!\left(\int_{\mathbb T^2}\rho(y)\,\mathrm{d}y\right),
\end{aligned}
\]
where the last inequality holds by Lemma~\ref{lem:fixed-deletion-covering} and the fact that $\sum_{i=0}^{N-1}\Delta_i(y)\le\rho(y)^2$, concluding the proof.

\section*{Acknowledgement}
{
The proof was discovered by GPT-5.6 Sol during an interactive process guided by the authors. The authors subsequently worked with GPT to develop and refine the arguments and exposition. The authors independently verified all mathematical claims and take full responsibility for the final manuscript.}

\bibliographystyle{alpha}
\bibliography{references}

@article{ajtai1984optimal,
  title={On optimal matchings},
  author={Ajtai, Mikl{\'o}s and Koml{\'o}s, J{\'a}nos and Tusn{\'a}dy, G{\'a}bor},
  journal={Combinatorica},
  volume={4},
  number={4},
  pages={259--264},
  year={1984},
  publisher={Springer-Verlag Berlin/Heidelberg}
}

@article{DBLP:journals/ior/YangY26,
  author       = {Mingwei Yang and
                  Sophie H. Yu},
  title        = {Online Metric Matching: Beyond the Worst Case},
  journal      = {Oper. Res.},
  volume       = {74},
  number       = {1},
  pages        = {130--140},
  year         = {2026}
}

@article{DBLP:journals/ipl/TsaiTC94,
  author       = {Ying The Tsai and
                  Chuan Yi Tang and
                  Yunn Yen Chen},
  title        = {Average Performance of a Greedy Algorithm for the On-Line Minimum
                  Matching Problem on Euclidean Space},
  journal      = {Inf. Process. Lett.},
  volume       = {51},
  number       = {6},
  pages        = {275--282},
  year         = {1994}
}

@article{talagrand1992matching,
  title={Matching random samples in many dimensions},
  author={Talagrand, Michel},
  journal={The Annals of Applied Probability},
  pages={846--856},
  year={1992},
  publisher={JSTOR}
}

@article{DBLP:journals/mor/SentenacNLMP26,
  author       = {Flore Sentenac and
                  Nathan Noiry and
                  Matthieu Lerasle and
                  Laurent M{\'{e}}nard and
                  Vianney Perchet},
  title        = {Online Matching in Geometric Random Graphs},
  journal      = {Math. Oper. Res.},
  volume       = {51},
  number       = {2},
  pages        = {1585--1625},
  year         = {2026}
}

@inproceedings{DBLP:conf/compgeom/Raghvendra18,
  author       = {Sharath Raghvendra},
  title        = {Optimal Analysis of an Online Algorithm for the Bipartite Matching
                  Problem on a Line},
  booktitle    = {SoCG},
  series       = {LIPIcs},
  volume       = {99},
  pages        = {67:1--67:14},
  publisher    = {Schloss Dagstuhl - Leibniz-Zentrum f{\"{u}}r Informatik},
  year         = {2018}
}

@article{DBLP:journals/talg/PesericoS23,
  author       = {Enoch Peserico and
                  Michele Scquizzato},
  title        = {Matching on the Line Admits no o({\(\surd\)}log n)-Competitive Algorithm},
  journal      = {{ACM} Trans. Algorithms},
  volume       = {19},
  number       = {3},
  pages        = {28:1--28:4},
  year         = {2023}
}

@inproceedings{DBLP:conf/soda/MeyersonNP06,
  author       = {Adam Meyerson and
                  Akash Nanavati and
                  Laura J. Poplawski},
  title        = {Randomized online algorithms for minimum metric bipartite matching},
  booktitle    = {{SODA}},
  pages        = {954--959},
  publisher    = {{ACM} Press},
  year         = {2006}
}

@article{DBLP:journals/tcs/KhullerMV94,
  author       = {Samir Khuller and
                  Stephen G. Mitchell and
                  Vijay V. Vazirani},
  title        = {On-Line Algorithms for Weighted Bipartite Matching and Stable Marriages},
  journal      = {Theor. Comput. Sci.},
  volume       = {127},
  number       = {2},
  pages        = {255--267},
  year         = {1994}
}

@article{kanoria2025dynamic,
  title={Dynamic spatial matching},
  author={Kanoria, Yash},
  journal={The Annals of Applied Probability},
  volume={35},
  number={5},
  pages={3086--3118},
  year={2025},
  publisher={Institute of Mathematical Statistics}
}

@article{DBLP:journals/jal/KalyanasundaramP93,
  author       = {Bala Kalyanasundaram and
                  Kirk Pruhs},
  title        = {Online Weighted Matching},
  journal      = {J. Algorithms},
  volume       = {14},
  number       = {3},
  pages        = {478--488},
  year         = {1993}
}

@article{holden2021gravitational,
  title={Gravitational allocation for uniform points on the sphere},
  author={Holden, Nina and Peres, Yuval and Zhai, Alex},
  year={2021}
}

@inproceedings{DBLP:conf/icalp/GuptaGPW19,
  author       = {Anupam Gupta and
                  Guru Guruganesh and
                  Binghui Peng and
                  David Wajc},
  title        = {Stochastic Online Metric Matching},
  booktitle    = {{ICALP}},
  series       = {LIPIcs},
  volume       = {132},
  pages        = {67:1--67:14},
  publisher    = {Schloss Dagstuhl - Leibniz-Zentrum f{\"{u}}r Informatik},
  year         = {2019}
}

@article{DBLP:journals/orl/GairingK19,
  author       = {Martin Gairing and
                  Max Klimm},
  title        = {Greedy metric minimum online matchings with random arrivals},
  journal      = {Oper. Res. Lett.},
  volume       = {47},
  number       = {2},
  pages        = {88--91},
  year         = {2019}
}

@article{bobkov2021simple,
  title={A simple Fourier analytic proof of the AKT optimal matching theorem},
  author={Bobkov, Sergey G and Ledoux, Michel},
  journal={The Annals of Applied Probability},
  volume={31},
  number={6},
  pages={2567--2584},
  year={2021},
  publisher={Institute of Mathematical Statistics}
}

@article{DBLP:journals/algorithmica/BansalBGN14,
  author       = {Nikhil Bansal and
                  Niv Buchbinder and
                  Anupam Gupta and
                  Joseph Naor},
  title        = {A Randomized O(log2 k)-Competitive Algorithm for Metric Bipartite
                  Matching},
  journal      = {Algorithmica},
  volume       = {68},
  number       = {2},
  pages        = {390--403},
  year         = {2014}
}

@inproceedings{DBLP:conf/sigecom/BalkanskiFP23,
  author       = {Eric Balkanski and
                  Yuri Faenza and
                  No{\'{e}}mie P{\'{e}}rivier},
  title        = {The Power of Greedy for Online Minimum Cost Matching on the Line},
  booktitle    = {{EC}},
  pages        = {185--205},
  publisher    = {{ACM}},
  year         = {2023}
}

@inproceedings{DBLP:conf/sigecom/AkbarpourALS22,
  author       = {Mohammad Akbarpour and
                  Yeganeh Alimohammadi and
                  Shengwu Li and
                  Amin Saberi},
  title        = {The Value of Excess Supply in Spatial Matching Markets},
  booktitle    = {{EC}},
  pages        = {62},
  publisher    = {{ACM}},
  year         = {2022}
}

@article{tong2016online,
  title={Online minimum matching in real-time spatial data: experiments and analysis},
  author={Tong, Yongxin and She, Jieying and Ding, Bolin and Chen, Lei and Wo, Tianyu and Xu, Ke},
  journal={Proceedings of the VLDB Endowment},
  volume={9},
  number={12},
  pages={1053--1064},
  year={2016},
  publisher={VLDB Endowment}
}

@inproceedings{DBLP:conf/innovations/LiV026,
  author       = {Yingxi Li and
                  Ellen Vitercik and
                  Mingwei Yang},
  title        = {Smoothed Analysis of Online Metric Matching with a Single Sample:
                  Beyond Metric Distortion},
  booktitle    = {{ITCS}},
  series       = {LIPIcs},
  volume       = {362},
  pages        = {94:1--94:23},
  publisher    = {Schloss Dagstuhl - Leibniz-Zentrum f{\"{u}}r Informatik},
  year         = {2026}
}

@inproceedings{sibson1980vector,
  title={A vector identity for the Dirichlet tessellation},
  author={Sibson, Robin},
  booktitle={Mathematical Proceedings of the Cambridge Philosophical Society},
  volume={87},
  number={1},
  pages={151--155},
  year={1980},
  organization={Cambridge University Press}
}

@article{du1999centroidal,
  title={Centroidal Voronoi tessellations: Applications and algorithms},
  author={Du, Qiang and Faber, Vance and Gunzburger, Max},
  journal={SIAM review},
  volume={41},
  number={4},
  pages={637--676},
  year={1999},
  publisher={SIAM}
}

@article{DBLP:journals/sigecom/HuangTW24,
  author       = {Zhiyi Huang and
                  Zhihao Gavin Tang and
                  David Wajc},
  title        = {Online Matching: {A} Brief Survey},
  journal      = {SIGecom Exch.},
  volume       = {22},
  number       = {1},
  pages        = {135--158},
  year         = {2024}
}

@inproceedings{DBLP:conf/stoc/KarpVV90,
  author       = {Richard M. Karp and
                  Umesh V. Vazirani and
                  Vijay V. Vazirani},
  title        = {An Optimal Algorithm for On-line Bipartite Matching},
  booktitle    = {{STOC}},
  pages        = {352--358},
  publisher    = {{ACM}},
  year         = {1990}
}

@article{DBLP:journals/jacm/MehtaSVV07,
  author       = {Aranyak Mehta and
                  Amin Saberi and
                  Umesh V. Vazirani and
                  Vijay V. Vazirani},
  title        = {AdWords and generalized online matching},
  journal      = {J. {ACM}},
  volume       = {54},
  number       = {5},
  pages        = {22},
  year         = {2007}
}

@book{tao2012higher,
  title={Higher order Fourier analysis},
  author={Tao, Terence},
  year={2012},
  publisher={American Mathematical Soc.}
}

@book{burago2001course,
  title     = {A Course in Metric Geometry},
  author    = {Burago, Dmitri and Burago, Yuri D. and Ivanov, Sergei},
  volume    = {33},
  year      = {2001},
  publisher = {American Mathematical Society}
}

@article{DBLP:journals/siamcomp/AigerKS14,
  author       = {Dror Aiger and
                  Haim Kaplan and
                  Micha Sharir},
  title        = {Reporting Neighbors in High-Dimensional Euclidean Space},
  journal      = {{SIAM} J. Comput.},
  volume       = {43},
  number       = {4},
  pages        = {1363--1395},
  year         = {2014}
}

@article{DBLP:journals/jacm/Arora98,
  author       = {Sanjeev Arora},
  title        = {Polynomial Time Approximation Schemes for Euclidean Traveling Salesman
                  and other Geometric Problems},
  journal      = {J. {ACM}},
  volume       = {45},
  number       = {5},
  pages        = {753--782},
  year         = {1998}
}

@article{DBLP:journals/ior/ChenKKZ26,
  author       = {Yilun Chen and
                  Yash Kanoria and
                  Akshit Kumar and
                  Wenxin Zhang},
  title        = {Feature-Based Dynamic Matching},
  journal      = {Oper. Res.},
  volume       = {74},
  number       = {2},
  pages        = {788--803},
  year         = {2026}
}

@article{DBLP:journals/corr/abs-2606-05546,
  author       = {Josh Ascher and
                  Eric Balkanski and
                  Jason Chatzitheodorou and
                  Vasilis Gkatzelis},
  title        = {Online Min-Cost Matching with General Arrivals},
  journal      = {CoRR},
  volume       = {abs/2606.05546},
  year         = {2026}
}

@article{kumar2026feature,
  title={Feature-based Dynamic Matching in the Dark},
  author={Kumar, Akshit},
  journal={Available at SSRN 6839679},
  year={2026}
}

@inproceedings{DBLP:conf/approx/Raghvendra16,
  author       = {Sharath Raghvendra},
  title        = {A Robust and Optimal Online Algorithm for Minimum Metric Bipartite
                  Matching},
  booktitle    = {{APPROX-RANDOM}},
  series       = {LIPIcs},
  volume       = {60},
  pages        = {18:1--18:16},
  publisher    = {Schloss Dagstuhl - Leibniz-Zentrum f{\"{u}}r Informatik},
  year         = {2016}
}

@article{steele1986efron,
  author  = {J. Michael Steele},
  title   = {An {Efron--Stein} Inequality for Nonsymmetric Statistics},
  journal = {The Annals of Statistics},
  volume  = {14},
  number  = {2},
  pages   = {753--758},
  year    = {1986},
  doi     = {10.1214/aos/1176349952}
}

@article{boucheron2005moment,
  author  = {St{\'e}phane Boucheron and Olivier Bousquet and G{\'a}bor Lugosi and Pascal Massart},
  title   = {Moment Inequalities for Functions of Independent Random Variables},
  journal = {The Annals of Probability},
  volume  = {33},
  number  = {2},
  pages   = {514--560},
  year    = {2005},
  doi     = {10.1214/009117904000000856}
}

@book{talagrand2022upper,
  title={Upper and lower bounds for stochastic processes: decomposition theorems},
  author={Talagrand, Michel},
  year={2022},
  publisher={Springer Nature}
}

\appendix
\newpage

\section{Omitted Proofs in Section~\ref{sec:preliminaries}}
\label{app:preliminaries-proofs}

\subsection{Proof of Fact~\ref{fact:torus-translation-stationarity}}
\label{app:proof-torus-translation}

Fix an input sequence whose Greedy choices are unique.
Fix $u\in\mathbb T^d$.  For $x,y\in\mathbb T^d$, translating both points by $u$ leaves their
coordinate differences modulo one unchanged, so it preserves
$d_{\mathbb T}(x,y)$.  The map $\tau_u$ also preserves uniform volume.
At the first request, distance preservation and uniqueness imply
that the Greedy process on the translated inputs selects the translated
copy of the server chosen on the original inputs.
After deleting these servers, the remaining configurations still differ
by the same translation.  Induction over the request sequence shows
that all Greedy choices are translated by $\tau_u$ and, hence the unmatched
configuration at every rank is translated by $\tau_u$.

Under iid uniform input, simultaneous translation leaves the joint
distribution of the servers and requests unchanged.  The preceding
identity for the unmatched configurations therefore implies that
$\tau_u(B_m)$ has the same distribution as $B_m$ for every $u$ and $m$.
To obtain the expected-count identity in the fact, let $W$ be a uniform
torus point independent of $B_m$.  Conditional on $B_m$, each translated
server $\tau_W(b)$ is uniform, so linearity of expectation gives
\[
  \E\bigl[|\tau_W(B_m)\cap P|\mid B_m\bigr]
  =\sum_{b\in B_m}\Pp\{\tau_W(b)\in P\mid B_m\}
  =m|P|.
\]
Since $\tau_W(B_m)$ has the same distribution as $B_m$, taking
expectations gives \eqref{eq:torus-stationarity}, completing the proof
of Fact~\ref{fact:torus-translation-stationarity}.

\subsection{Proof of Fact~\ref{lem:random-cut}}
\label{app:proof-random-cut}

Note that unwrapping cannot decrease distances: for every cut $U$ and every
$z,z'\in\mathbb T^d$,
\begin{equation}\label{eq:metric-domination}
  d_{\mathbb T}(z,z')\le d_U(z,z').
\end{equation}
Indeed, $\kappa_U(z)-\kappa_U(z')$ differs from
$z-z'$ by an integer vector, so its norm is at least
$d_{\mathbb T}(z,z')$.\footnote{See also \cite[Lemma~3.3.6]{burago2001course}.}
If none of the chosen shortest coordinate arcs is cut, the
unwrapped coordinate differences realize these arcs, so
$d_U(z,z')=d_{\mathbb T}(z,z')$.

For Property~\ref{item:fact-cut-item-v}, each uniform cut point $U_i$
lies on the chosen coordinate arc with probability $\Delta_i(z,z')$.
The union bound and the Cauchy--Schwarz inequality therefore give
\[
  \Pp_U\{\text{at least one chosen arc is cut}\}
  \le\sum_{i=1}^d\Delta_i(z,z')
  \le\sqrt d\,d_{\mathbb T}(z,z'),
\]
proving Property~\ref{item:fact-cut-item-v}.\footnote{See
also~\cite[Sections~1--2]{DBLP:journals/siamcomp/AigerKS14}
for the shifted-grid bound.}

It remains to check the nearest-neighbor implication in
Property~\ref{item:fact-cut-item-vi}.
Let $b$ be a torus nearest neighbor of $x$ in $S$, and suppose
none of the chosen shortest coordinate arcs from $x$ to $b$ is cut.
For every $z\in S\setminus\{b\}$,
\[
  d_U(x,z)\ge d_{\mathbb T}(x,z)
             \geq d_{\mathbb T}(x,b)
             =d_U(x,b),
\]
where the first inequality holds by \eqref{eq:metric-domination}, the second inequality holds by the nearest-neighbor property of $b$, and the equality holds since none of the chosen coordinate arcs from $x$ to $b$ is cut.
Therefore, $b$ remains a nearest neighbor of $x$ under the cut metric, proving
Property~\ref{item:fact-cut-item-vi}.

\subsection{Proof of Lemma~\ref{lem:random-grid-separation}}
\label{app:proof-random-grid-separation}

The boundary-crossing estimates for randomly shifted dissections are
developed in~\cite[Sections~2.2--2.3]{DBLP:journals/jacm/Arora98}.
For the precise Euclidean grid bound used here, the discussion of the first algorithm in
\cite[Section~1]{DBLP:journals/siamcomp/AigerKS14} gives the Euclidean
separation bound $\min\{1,\sqrt d\,\|x-y\|_2/s\}$ for a randomly shifted
grid of side length $s$.
For torus points $x,y$, fix periodic copies
$\tilde x,\tilde y\in\mathbb R^d$, independently of $V$, such that
$\|\tilde x-\tilde y\|_2=d_{\mathbb T}(x,y)$.
Use the same shift $V$ for the Euclidean and torus grids.
Because $1/s$ is an integer, translating a Euclidean cell by any
integer vector gives another cell of the same grid, and reduction modulo
one identifies these cells.
In particular, if $\tilde x,\tilde y$ lie in the same Euclidean
cell, then $x,y$ lie in the same torus cell.
Taking the contrapositive gives the event inclusion
\[
  \{x,y\text{ lie in different torus cells}\}
  \subseteq
  \{\tilde x,\tilde y\text{ lie in different Euclidean cells}\}.
\]
Therefore, the same cited estimate applies, and $s\in[h/2,h]$ gives
\eqref{eq:random-grid-separation}.

\subsection{Proof of Lemma~\ref{lem:periodic-voronoi-deletion}}
\label{app:proof-periodic-deletion}

Fix a finite set $S\subset\mathbb T^d$ with $|S|\ge2$ and $s\in S$.
Let $C:=[-1/2,1/2]^d$.
Define the following potential on $\mathbb R^d$:
\[
\psi_s(x):=d_{\mathbb T}(s+x,S\setminus\{s\})^2
  -d_{\mathbb T}(s+x,S)^2,
\]
which is nonnegative and periodic.
For every $x\in C$, the origin is a nearest periodic copy of $s$.
Hence, for $x\in C\setminus\mathcal V_s$, deleting $s$ leaves the distance from $x$ to its nearest server unchanged,
implying $\psi_s(x)=0$.
Since $\mathcal V_s$ represents the torus Voronoi
cell of $s$,
\begin{equation}\label{eq:potential-increment}
  \Delta_sQ(S)=\int_{\mathcal V_s}\psi_s(x)\,\mathrm{d}x
  =\int_C\psi_s(x)\,\mathrm{d}x.
\end{equation}
Moreover, for $x\in\mathcal V_s$, since the origin is a nearest server copy
before deleting $s$, we have
\begin{align}
  \psi_s(x)
  =\|x-T_s(x)\|_2^2-\|x\|_2^2
  =
  \min_{\substack{t\in S\setminus\{s\}\\z\in\mathbb Z^d}}
  \left\{
    \|t-s+z\|_2^2-2x\cdot(t-s+z)
  \right\}.
  \label{eq:psi-affine}
\end{align}
It is not hard to see that $\psi_s(x)$ is continuous and piecewise affine on $C$, with
gradient $-2T_s$ almost everywhere on $\mathcal V_s$ and zero almost
everywhere on $C\setminus\mathcal V_s$.

\paragraph{Zero drift.}
Since $\nabla \psi_s = -2 T_s$ almost everywhere on $\mathcal V_s$, it holds that
\[
  \int_{\mathcal V_s}T_s(x)\,\mathrm{d}x
  =-\frac12\int_{C}\nabla\psi_s(x)\,\mathrm{d}x
  =0,
\]
where the second equality holds by the periodicity of $\psi_s(x)$, concluding the zero-drift identity.

\paragraph{Bound on the second moment.}
By \eqref{eq:psi-affine} and $\nabla \psi_s = -2 T_s$,
\begin{align*}
    \|T_s(x)\|_2^2 =\psi_s(x)-x\cdot\nabla\psi_s(x)
\end{align*}
almost everywhere on $\mathcal V_s$.
Let $\nu(\cdot)$ be the outward unit normal to $C$, and let $\mathrm{d}\sigma$ be surface area on the boundary of $C$.
It follows that
\[
\begin{aligned}
  \int_{\mathcal V_s}\|T_s(x)\|_2^2\,\mathrm{d}x
  &=\int_{\mathcal V_s}\bigl(\psi_s(x)-x\cdot\nabla\psi_s(x)\bigr)\,\mathrm{d}x  \\
  &=\int_C\bigl(\psi_s(x)-x\cdot\nabla\psi_s(x)\bigr)\,\mathrm{d}x\\
  &=(d+1)\Delta_sQ(S)
    -\int_{C}\nabla\cdot(x\psi_s(x))\,\mathrm{d}x\\
  &=(d+1)\Delta_sQ(S)
    -\int_{\partial C}\psi_s(x)x\cdot\nu(x)\,\mathrm{d}\sigma\\
  &=(d+1)\Delta_sQ(S)
    -\frac12\int_{\partial C}\psi_s(x)\,\mathrm{d}\sigma\\
  &\le(d+1)\Delta_sQ(S).
\end{aligned}
\]
The second equality holds since $\psi_s-x\cdot\nabla\psi_s$ vanishes almost everywhere
on $C\setminus\mathcal V_s$.
The third equality holds by \eqref{eq:potential-increment} and the product rule $\nabla\cdot(x\psi_s)=d\psi_s+x\cdot\nabla\psi_s$.
For the fourth equality, we partition $C$ into finitely many polytopes on which $\psi_s$ is affine and apply the divergence theorem to each polytope, where the boundary integrals over each shared
interior face cancel.
The last equality holds since $x\cdot\nu(x)=1/2$ almost everywhere on $\partial C$.
The inequality holds since $\psi_s$ is nonnegative.
\section{Omitted Proofs in Section~\ref{sec:greedy-properties}}
\label{app:greedy-proofs}

\subsection{Proof of Corollary~\ref{cor:torus-mean-costs}}
\label{app:proof-torus-mean-costs}

To bound the expected total cost, we integrate the tail bound at each rank, using the scale $\ell_m$ from \eqref{eq:intro-rank-scale},
and then sum over all ranks.
For $d\ge2$, integrating \eqref{eq:torus-tail} above $\ell_m$ gives
\begin{equation}\label{eq:torus-edge-integration}
\E D_m^{\mathbb T}
  =O\!\left(\ell_m+\ell_m^d\int_{\ell_m}^1r^{-d}\,\mathrm{d}r\right)
  =O(\ell_m).
\end{equation}
To see the last equality above, note that $\ell_m^d=n/m^2$ when $m>\sqrt n$, and hence $\int_\ell^1r^{-d}\,\mathrm{d}r\le\ell^{1-d}/(d-1)$.
In addition, when $m\le\sqrt n$, we have $\ell_m=1$, and the same bound follows
from the torus diameter.

For $d=1$ and $m\le\sqrt n$, the torus diameter gives
$\E D_m^{\mathbb T}\le1$.  For $m>\sqrt n$, integrating \eqref{eq:torus-tail} gives
\begin{equation}\label{eq:torus-edge-one}
\E D_m^{\mathbb T}
  =O\!\left(\ell_m+\ell_m\int_{\ell_m}^1\frac{\mathrm{d}r}{r}\right)
  =O\!\left(\frac{n}{m^2}\left(1+\log\frac{m^2}{n}\right)\right).
\end{equation}

It remains to sum the bounds over all ranks.
\begin{equation*}
\sum_{k=1}^n\ell_k
  =O\!\left(\sqrt n+n^{1/d}\int_{\sqrt n}^{n}x^{-2/d}\,\mathrm{d}x\right)
  =O(M_n).
\end{equation*}
For $d\ge2$, summing \eqref{eq:torus-edge} and applying this estimate
proves \eqref{eq:torus-total-one}.  For $d=1$, summing
\eqref{eq:torus-edge-one} and using the diameter bound for $m\le\sqrt n$
gives
\[
  \sum_{m=1}^n\E D_m^{\mathbb T}
  =O\!\left(\sqrt n+n\int_{\sqrt n}^{\infty}
        \frac{1+\log(x^2/n)}{x^2}\,\mathrm{d}x\right)
  =O(\sqrt n),
\]
concluding the proof.

\subsection{Proof of Lemma~\ref{lem:unwrapped-greedy-law}}
\label{app:proof-unwrapped-greedy-law}

Condition on $U=u$.
Since $U$ is independent of the inputs and $\kappa_u$ preserves
uniform measure\footnote{See, e.g., \cite[Section~1.1]{tao2012higher}.},
the transformed servers and requests are iid uniform in $[0,1)^d$.
By the identity
$d_u(x,y)=\|\kappa_u(x)-\kappa_u(y)\|_2$, $\kappa_u$ preserves
each nearest-server choice and its matching cost.
Induction over the requests therefore identifies the entire
transformed process with Euclidean Greedy on these inputs, proving
the asserted conditional joint law.

\end{document}